%% file: main.tex
\documentclass[runningheads]{llncs}
\usepackage[T1]{fontenc}
\usepackage{graphicx}
\usepackage{mathpartir}
\usepackage{amsmath}
\usepackage{amssymb}
\usepackage{listings}
\usepackage{todonotes}
\usepackage{lstautogobble}
\usepackage{ifsym}
\usepackage{subcaption}
\usepackage{hyperref}
\usepackage[ruled, vlined, linesnumbered]{algorithm2e}
\usepackage{cleveref}
\usepackage{xspace}
\usepackage{stmaryrd}
\usepackage{orcidlink}
\usepackage{ifthen}
\usepackage{mathtools}
\usepackage{wrapfig}
\usepackage{marvosym}

\newboolean{fullversion}
\setboolean{fullversion}{true}

\usetikzlibrary{backgrounds,arrows.meta}

\input{macros.tex}
\input{lola-syntax.tex}

\begin{document}
%
\title{Two Ways to See the Future: Combining Prediction and Future-Offset Accesses in RTLola}
\titlerunning{Combining Prediction and Future-Offset Accesses in RTLola}
%
\author{Jan Baumeister\inst{1}\orcidlink{0000-0002-8891-7483} \and
Bernd Finkbeiner\inst{1,2}\orcidlink{0000-0002-4280-8441} \and
Eduard Müller\inst{1}\orcidlink{0009-0007-1511-2469} \and\\
Frederik Scheerer\textsuperscript{(\Letter)}\inst{1}\orcidlink{0009-0007-8115-0359} \and
Julia Tillman\inst{3}
}
\authorrunning{Baumeister et. al.}
%
\institute{
CISPA Helmholtz Center for Information Security, Saarbrücken, Germany
\email{\{jan.baumeister,finkbeiner,eduard.mueller,frederik.scheerer\}@cispa.de} \and
Technical University of Munich, Germany \and
Saarland University, Saarbrücken, Germany
}
\maketitle              
\begin{abstract}
RTLola is a stream-based specification language designed for asynchronous real-time systems.
While many temporal specifications naturally refer to future behavior, RTLola currently offers no mechanism to express such future-dependent properties.
In this paper, we extend RTLola with two complementary mechanisms to reason about the future.
First, we introduce a \emph{prediction operator} that extrapolates future stream values at arbitrary timestamps based on past observations.
Second, we add a discrete \emph{future offset operator}, which provides access to precise future values by delaying the evaluation of the dependent stream expressions.
While the former enables immediate, but possibly imprecise predictions, the latter ensures exact values once the required information becomes available.
We formalize both extensions in the RTLola semantics and evaluate their implementation on runtime and memory consumption.

\end{abstract}

\input{content/intro.tex}
\input{content/preliminaries.tex}

\input{content/predict.tex}
\input{content/offsets.tex}
\input{content/evaluation.tex}
\input{content/conclusion.tex}

\begin{credits}
\subsubsection{\ackname} This work was partially supported by the German Research Foundation (DFG) as part of TRR 248 (No.~389792660) and PreCePT (No.~521273327), and by the European Research Council (ERC) Grant HYPER (No.~101055412).

\subsubsection{\discintname}
The authors have no competing interests to declare that are relevant to the content of this article.
\end{credits}

\bibliographystyle{splncs04}
\bibliography{bibliography.bib}

\ifthenelse{\boolean{fullversion}}{
    \clearpage
    \appendix
    \input{content/appendix.tex}
}{}

\end{document}

%% file: macros.tex
\newcommand{\NN}{\mathbb{N}}
\newcommand{\RR}{\mathbb{R}}

\newcommand{\ZZ}{\mathbb{Z}}
\newcommand{\VV}{\mathbb{V}}

\newcommand{\powerset}[1]{\mathcal{P}\left(#1\right)}
\newcommand{\FRTLola}{\textsc{RTLola}$^f$\xspace}
\newcommand{\RTLola}{\textsc{RTLola}\xspace}
\newcommand{\Lola}{\textsc{Lola}\xspace}
\newcommand{\Iref}{\mathtt{ID}^\uparrow}
\newcommand{\Oref}{\mathtt{ID}^\downarrow}

\newcommand{\EG}[2][\varphi]{EG_{#1}(#2)}
\newcommand{\sid}{\mathit{sid}}
\newcommand{\world}{\omega}
\newcommand{\World}{\mathbb{W}}
\newcommand{\Time}{\text{Time}}

\newcommand{\semantics}[1]{\left\llbracket#1\right\rrbracket}

\newcommand{\offset}[1]{\mathit{offset}(\mathit{by}: #1)}
\newcommand{\rtoffset}[1]{\mathit{offset}_{RT}(\mathit{by}: #1)}

%% file: lola-syntax.tex
\definecolor{bluekeywords}{rgb}{0.13, 0.13, 1}
\definecolor{greentypes}{rgb}{0, 0.5, 0}
\definecolor{orangecomments}{rgb}{1, 0.5, 0.1}
\definecolor{redstrings}{RGB}{171, 114, 2}
\definecolor{graynumbers}{rgb}{0.5, 0.5, 0.5}
\definecolor{goldcomments}{rgb}{0.6, 0.4, 0.08}

\definecolor{implemented}{rgb}{0.67, 0.9, 0.93}
\colorlet{existing}{lightgray}

\lstdefinelanguage{Lola}{
  keywords=[0]{input, output, trigger, constant, import, spawn, eval, close, with, when, filter},
  moredelim=**[is][\transparent{0.6}]{?}{?},
  moredelim=**[is][\color{greentypes}@]{@}{@},
  keywordstyle=[0]\bfseries\color{bluekeywords},
  keywords=[1]{if, then, else, aggregate, defaults, offset, last, by, or, to, sin, cos, abs, hold, over, using, over_instances, prob, over_discrete, in, predict},
  keywords=[2]{Variable, String, Int, Int64, UInt, UInt64, Bool, Float32, Float64, Float},
  keywordstyle=[2]\color{greentypes},
  sensitive=false,
  comment=[l]{//},
  morecomment=[s]{/*}{*/},
  morestring=[b]',
  morestring=[b]",
  literate={\\@}{@}1,
  moredelim=[s][\color{purple}\bfseries]{\#[}{]},
}
\makeatletter

%% file: content/intro.tex
\section{Introduction}

Many monitoring tasks naturally require reasoning about future system behavior.
Such requirements can be expressed naturally in future-time temporal logic:
for instance, the requirement that each takeoff event must be eventually followed by the aircraft reaching a valid flight altitude can be written in LTL as
\[
\square(\mathit{takeoff} \Rightarrow \bigcirc(\neg\mathit{takeoff} \;\mathcal{U}\; \mathit{alt\_reached})).
\]
In monitoring, specifications that refer to the past are usually preferable to those that reason about the future, because past-time specifications can be evaluated in an online manner without lookahead. For the monitorable fragments of logics like LTL, it is possible to translate future-time specifications into equivalent past-time specifications at the cost of an exponential blow-up in the formula~\cite{DBLP:conf/lop/LichtensteinPZ85,DBLP:conf/time/ArtaleGGMM23}.
For stream-based specification languages like \RTLola~\cite{DBLP:conf/fm/BaumeisterFKS24}, the situation is more difficult.
Unlike temporal logics, which evaluate a single verdict for an entire trace, \RTLola produces a new value at every time point.
A future-to-past translation would shift the evaluation points of the stream expressions, and with it, the timestamp at which a violation is reported.
However, in a real-time monitoring context, the precise time of a detected event is important.
For example, we might want the monitor to report the timestamp of the takeoff that triggered the violation, not the timestamp at which the violation was confirmed.

RTLola currently offers no mechanism to
express future-time specifications.
In this paper, we introduce 
two complementary mechanisms for accessing future information.
A \emph{prediction operator} for getting immediate, but possibly imprecise predictions, and a \emph{future-access operator}, which gets exact, but delayed information.
The prediction operator enables forecasting of future stream values based on a bounded history of past observations.
At each step, a prediction function is applied to a sliding window of recent values to estimate the stream’s value at a future offset.
This prediction function can be instantiated with standard models such as linear regression or other statistical and machine learning techniques.
This design enables flexible integration of domain-specific forecasting methods while preserving a uniform interface at the specification level.
However, such predictions are inherently uncertain.
In highly dynamic scenarios, the predicted values may deviate significantly from the actual future values.

An alternative approach, which was already present in the original \Lola~\cite{DBLP:conf/time/DAngeloSSRFSMM05} language, is to access future values directly via a future-offset operator.
This operator refers to values that will occur at a later point in time, providing exact information rather than estimates.
Naturally, this comes at the cost of a delayed evaluation: results can only be produced once the referenced future values have been observed.
Consequently, specifications using future offsets introduce latency but yield precise results.
This future access operator complements the prediction operator.
The two approaches provide a trade-off between immediacy and accuracy, and can be combined to obtain early estimates and precise confirmations.
In contrast to \Lola, supporting future-offset semantics in \RTLola is more difficult due to its asynchronous stream model, where streams may evolve at different rates and lack a global synchronized clock.
As a consequence, future references are not sound in general.
We identify a well-defined subset of \RTLola specifications for which future-offset accesses are supported, and strengthen the language's well-formedness criteria to guarantee this property.
We evaluate our approach on runtime and memory usage, and show that both remain manageable and, in many cases, comparable to standard \RTLola monitoring.
Finally, we show that depending on the input data, our implementation can maintain bounded memory even for specifications that are unbounded in principle.

The remainder of this paper is structured as follows.
\Cref{sec:preliminaries} introduces the necessary background on \RTLola.
\Cref{sec:predict} presents the prediction operator, followed by the future offset operator in \Cref{sec:future_offset}.
\Cref{sec:evaluation} then reports on the experimental evaluation.

\subsection{Related Work}
\enlargethispage*{\baselineskip}

Past-time and future-time specifications represent two established approaches to runtime monitoring~\cite{DBLP:conf/tacas/HavelundR02,DBLP:journals/tosem/BauerLS11}.
While past-time operators enable incremental evaluation, future-time operators allow more natural and succinct specifications~\cite{DBLP:conf/time/ArtaleGGMM23}, at the cost of delayed~\cite{DBLP:conf/time/DAngeloSSRFSMM05} or inconclusive verdicts~\cite{DBLP:journals/tosem/BauerLS11}. 
Our approach extends stream-based runtime monitoring with two mechanisms for reasoning about future behavior: prediction and exact future accesses.

Prediction of future system behavior has been studied extensively in runtime monitoring and cyber-physical systems.
Existing approaches rely on probabilistic and machine learning techniques, such as Bayesian estimation~\cite{DBLP:conf/iros/ChouY020}, neural network-based predictors~\cite{DBLP:conf/rv/CairoliBP21}, and general machine learning methods~\cite{DBLP:journals/inffus/ChenMLWL23}.
These methods typically require training data or domain specific knowledge.
In contrast, lightweight statistical techniques like polynomial regression~\cite{Spinger:article/bmc/FilipowMTRDDLS23} enable forecasting without extensive training.
Our work follows this approach, integrating predictive reasoning directly into the specification language \RTLola.
Predictive semantics for runtime verification have also been studied in the context of anticipating future property satisfaction or violation based on partial observations and system models~\cite{DBLP:conf/rv/Leucker12,DBLP:conf/nfm/ZhangLD12,DBLP:conf/iccps/LindemannQDP23,DBLP:conf/cav/HiplerKLS24}.

Besides prediction, several monitoring frameworks support exact future references.
In monitoring temporal logics, future-time operators are commonly used to reason about future system behavior~\cite{DBLP:journals/tosem/BauerLS11,DBLP:journals/jacm/BasinKMZ15,DBLP:conf/rv/SchumannMR15}.
Stream-based languages provide richer verdicts.
In this setting, future accesses appear mostly in monitoring languages for synchronous systems~\cite{DBLP:conf/time/DAngeloSSRFSMM05,DBLP:conf/rv/GorostiagaS18,DBLP:journals/sttt/GorostiagaS21} or in offline monitoring settings~\cite{DBLP:conf/rv/BozzelliS14,DBLP:conf/atva/RaszykBKT19}, where the complete trace is available during evaluation.
One example for a synchronous stream-based language is \Lola~\cite{DBLP:conf/time/DAngeloSSRFSMM05}, the predecessor of \RTLola, which supports future-offset accesses in a synchronous setting with a global clock.
In contrast, asynchronous stream-based languages~\cite{DBLP:conf/fm/BaumeisterFKS24,DBLP:conf/sbmf/ConventHLS0T18} allow streams to evolve independently at different rates.
While the asynchronous monitoring framework TeSSLa~\cite{DBLP:conf/sbmf/ConventHLS0T18,DBLP:phd/dnb/Scheffel22} supports scheduling events at future timestamps via delay, it does not provide explicit future-offset accesses, as considered in this work.
Until now, the \RTLola~\cite{DBLP:conf/fm/BaumeisterFKS24} semantics are restricted to evaluation based only on present and past information.
Our work extends future-offset accesses from the synchronous setting of \Lola to the asynchronous setting of \RTLola. 

%% file: content/preliminaries.tex
\section{Background}\label{sec:preliminaries}

\RTLola~\cite{DBLP:conf/fm/BaumeisterFKS24} is a stream-based monitoring language.
Stream-based monitors operate over \emph{streams}, i.e., infinite sequences of values.
\emph{Input streams} continuously receive values of the monitored system.
\emph{Output streams} are defined through stream equations and derive new values by filtering and aggregating the data.
\emph{Trigger}, special boolean-valued output streams, indicate a specification violation.

\begin{figure}[t]
\begin{lstlisting}
input alt : Float64
output alt$\_$above_ground @alt@ := alt - GROUND_LEVEL
output alt$\_$smoothed @1s@ :=
    alt$\_$above$\_$ground.aggregate(over: 2s, using: avg).defaults(to: 0.0)
output alt$\_$reached @1s@ := alt$\_$smoothed > 10.0
output alt$\_$sustained @1s@ :=
    alt$\_$reached $\land$ alt$\_$reached.offset(by: -1).defaults(to: false)
trigger alt_sustained "altitude sustained"
\end{lstlisting}
\caption{
    An \RTLola specification monitoring an altitude sensor.
}
\label{fig:example_rtlola}
\end{figure}

Consider the specification in \Cref{fig:example_rtlola} as an example.
The input stream \lstinline!alt! provides the readings of a drone's altitude sensor.
First, this measurement is converted to altitude above ground level by subtracting the ground-level height.
To mitigate sensor noise, the output stream \lstinline!alt$\_$smoothed! computes a sliding-window average over the previous two seconds, evaluated periodically every second.
The stream \lstinline!alt$\_$reached! evaluates whether the smoothed altitude exceeds a threshold of 10.0, and \lstinline!alt$\_$sustained! checks whether this condition holds for two consecutive evaluations.
To achieve this, the stream accesses the previous value of \lstinline!alt$\_$reached! using \lstinline!offset!.
Since no previous value exists at the beginning of monitoring, a default value is provided for this case.
The trigger emits a notification once the condition is satisfied.

\subsection{Types}
Every stream in \RTLola has two types.
The \emph{value type} describes the kind of data contained in the stream, such as a boolean, integer, or floating-point number.
The \emph{pacing type} determines \emph{when} new stream values are computed.
\RTLola distinguishes between two kinds of pacing types.
\emph{Event-based streams} are evaluated whenever specific inputs receive new values, whereas \emph{periodic streams} are evaluated at a fixed frequency.
Both value and pacing type are often inferred automatically, but can also be annotated explicitly.
For example, the annotation \lstinline!@1s! on the stream \lstinline!alt$\_$smoothed! indicates that it is evaluated once per second, while the \lstinline!alt$\_$above$\_$ground! stream is annotated with the input stream \lstinline!@alt@!, s.t. it is evaluated for every new arriving altitude reading.

Further, pacing types can include dynamic filter conditions.
Assuming an additional boolean input \lstinline!alt$\_$reliable!, we can adapt the \lstinline!alt$\_$above$\_$ground! stream to only produce values when the current altitude reading is reliable:
\begin{lstlisting}
output alt_above_ground
    eval @alt$\textcolor{greentypes}{\land}$alt_reliable@ when alt_reliable with alt - GROUND_LEVEL
\end{lstlisting}
Only if the \lstinline!when! condition is satisfied, stream values are produced by evaluating the \lstinline!with! expression.
Notice that the \lstinline!:=! syntax used before is just syntactic sugar for streams with a \lstinline!when! condition of \lstinline!true!.

In \RTLola, we differentiate between synchronous and asynchronous stream accesses.
All accesses of stream values encountered so far are \emph{synchronous}, meaning that the pacing types guarantee that, whenever a stream is accessed, the accessed stream has a value available at exactly that point in time.
For instance, the \lstinline!alt_sustained! stream can access the \lstinline!alt_reached! stream synchronously because both streams are evaluated once per second.
In contrast, to make stream accesses across different-paced streams, we must use an \emph{asynchronous} access.
As an example, consider the \lstinline!alt_display! stream, which outputs the most recent value from the \lstinline!alt_above_ground! stream every second:
\begin{lstlisting}
    output alt_display @1s@ := alt_above_ground.hold().defaults(to: 0.0)
\end{lstlisting}
This stream makes use of an \lstinline!hold! access, which accesses the \emph{most recent value} of the target stream, regardless of when that value was produced.
Because this value could potentially not exist, a default value is required for this case.
Sliding windows are another example of asynchronous stream accesses, such as on the \lstinline!alt_smoothed! stream from \Cref{fig:example_rtlola}.
For a more detailed introduction to the \RTLola specification language, we refer the reader to the \RTLola Tutorial\footnote{\url{https://rtlola.cispa.de/playground/tutorial}}~\cite{DBLP:conf/fm/BaumeisterFKS24}.

\subsection{Semantics}
The semantics of \RTLola are defined with respect to an \emph{evaluation model}.
We consider discrete timesteps $\Time = \{n \in \NN \mid n \le T_{\mathit{max}}\}$ and a domain of optional stream values $\VV_{\bot} = \VV \cup \{\bot\}$.
\begin{definition}[Evaluation Model~\cite{DBLP:conf/fm/BaumeisterFKS24,DBLP:conf/tacas/BaumeisterFSSW25}] \label{def:evaluation_model}
    Let $\varphi$ be an \RTLola specification with a set of input stream references $\Iref$ and output stream references $\Oref$. An evaluation model $\world \in \World$ is defined as:
    \begin{align*}
        \mathit{Stream} &:= \Time \rightarrow \VV_{\bot}\\
        \mathit{StreamMap} &:=  \Iref \uplus \Oref \rightarrow \mathit{Stream}\\
        \mathit{TimeMap} &:= \Time \rightarrow \mathbb{R}^+\\
        \World &:= \mathit{StreamMap} \times \mathit{TimeMap}
    \end{align*}
\end{definition}

An evaluation model consists of a stream map and a time map.
The stream map assigns each stream reference a time-indexed sequence of values.
As \RTLola is asynchronous, values may be absent at certain timesteps, denoted by $\bot$.
The time map assigns discrete timestamps to real-time timestamps.

Given an evaluation model $\world = (\mathit{streams},\mathit{times}) \in \World$, we use the following shorthand notation:
$\world(t) := \mathit{times}(t)$ to return the real-time timestamp for a time $t \in \Time$, and $\world(s) := \mathit{streams}(s)$ returns the stream associated with stream reference $s \in \Iref \uplus \Oref$.
\begin{definition}[RTLola Semantics~\cite{DBLP:conf/rv/BaumeisterFS25}]
    Given an \RTLola specification $\varphi$, we define the semantics as the set:
    \begin{align*}
        \semantics{\varphi} = \{\world \in \World \; \vert \; \forall sid \in \Oref . \; &\forall t \in \Time . \; \varphi(sid) \Downarrow_{\world}^{t} \world(sid)(t) \\
        &\land \forall t \in \Time . \; \world(t) < \world(t+1)\}
    \end{align*}
\end{definition}
In the formula, $\varphi(sid) \Downarrow_{\world}^{t} v$ denotes that the defining stream equation $\varphi(sid)$ of stream $sid$, evaluated at time $t$ relative to the evaluation model $\world$, yields the value $v$.
This evaluation first checks whether the stream's pacing condition is satisfied at time $t$.
If so, the result is the value obtained from evaluating the stream's equation, and otherwise the result is $\bot$.
This definition ensures that each output stream calculates its value according to the stream's equations $\varphi(sid)$ at every timestep $t$.
Furthermore, the real-time timestamps must be strictly monotonically increasing.
We additionally require that specifications are well-defined:
\begin{definition}[Well-Defined Specification~\cite{DBLP:phd/dnb/Schwenger22,DBLP:conf/rv/BaumeisterFS25}]
An \RTLola specification $\varphi$ is \emph{well-defined}, if, for every inputs stream assignment $I : \Time \rightarrow  (\Iref \rightarrow \VV_\bot)$ and time map $T: \Time \rightarrow \RR^+$, there exists exactly one evaluation model $\world \in \semantics{\varphi}$ such that for all $t \in \Time$,
\[
\world(t) = T(t) \quad \text{and} \quad \forall sid \in \Iref . \; \world(sid)(t) = I(t)(sid).
\]
\end{definition}
In other words, a well-defined specification admits exactly one evaluation model for every possible input stream assignment.
Consequently, all output stream values are uniquely determined by the input values and their timing.

\subsection{Static Analysis}

\RTLola guarantees well-definedness by implementing a static analysis based on a dependency graph, which captures temporal dependencies between streams.
\begin{definition}[Dependency Graph~\cite{DBLP:phd/dnb/Schwenger22}]\label{def:dependency_graph}
    Let $\varphi$ be an \RTLola specification.
    The corresponding dependency graph of $\varphi$ is a directed multi-graph $G_\varphi=\langle V, E \rangle$ with $V = \Iref \uplus \Oref$.
    The set of edges is represented by $E \subseteq V \times L_D \times V$, where
    \begin{align*}
    L_D &= \{ \mathit{Filter}, \mathit{Eval} \} \times \left(\{ \mathit{Sync} \}
    \cup \{ \mathit{Hold} \}
    \cup (\{ \mathit{Offset} \} \times \ZZ)
    \cup (\{ \mathit{Aggr} \} \times \mathcal{F} \times \RR)\right)
    \end{align*}
    is the corresponding dependency label.
    The first component is the behavioral argument, and the second component is the stream access kind.
    An edge $(s, (b,k), s')$ is in $E$ iff the stream expression of $s$ contains a stream access to $s'$ of kind $k$.
\end{definition}

\Cref{fig:dep_graph_example} depicts the dependency graph for the specification in \Cref{fig:example_rtlola}.
Each stream access in the specification corresponds to an edge in the graph.
Since all accesses occur in \lstinline!with!-expressions, every edge is annotated with $\mathit{Eval}$.
Accesses in \lstinline!when!-conditions would carry a $\mathit{Filter}$ annotation instead.

The static analysis uses the dependency graph and is conducted along three dimensions.
The \emph{well-formedness} is a static criterion ensuring a well-defined specification.
The \emph{evaluation order} represents a partial order on streams and ensures that streams are efficiently resolved within one evaluation cycle.
Intuitively, this ensures that all dependencies of a stream are evaluated before the stream itself within each evaluation cycle.
Finally, the \emph{memory bound} of a stream gives us the number of values that must be retained for monitoring.
This computation ensures that the monitor operates within finite memory constraints.
The existing definitions~\cite{DBLP:phd/dnb/Schwenger22} are insufficient for the support of our proposed operators and are extended throughout the rest of this paper.

%% file: content/predict.tex
\section{Prediction Operator}\label{sec:predict}

\begin{figure}[t]
    \begin{minipage}{0.37\linewidth}
    \vspace*{4mm}
    \centering
    \begin{tikzpicture}[
        n/.style={font=\scriptsize,draw},
        node distance=5mm,
        l/.style={font=\scriptsize,right},
        shorten >=0.5mm
    ]
        \node[n] (alt_sustained) {\lstinline!alt_sustained!};
        \node[n,above=of alt_sustained] (alt_reached) {\lstinline!alt_reached!};
        \node[n,above=of alt_reached] (alt_smoothed) {\lstinline!alt_smoothed!};
        \node[n,above=of alt_smoothed] (alt_above_ground) {\lstinline!alt_above_ground!};
        \node[n,above=of alt_above_ground] (alt) {\lstinline!alt!};
    
        \draw[->] (alt_sustained) to[bend left] node[l,left] {$(\mathit{Eval}, \mathit{Sync})$} (alt_reached);
        \draw[->] (alt_sustained) to[bend right] node[l] {$(\mathit{Eval}, (\mathit{Offset}, -1))$} (alt_reached);
        \draw[->] (alt_reached) -- node[l] {$(\mathit{Eval}, \mathit{Sync})$} (alt_smoothed);
        \draw[->] (alt_smoothed) -- node[l] {$(\mathit{Eval}, (\mathit{Aggr}, \mathit{avg}, 2s))$} (alt_above_ground);
        \draw[->] (alt_above_ground) -- node[l] {$(\mathit{Eval}, \mathit{Sync})$} (alt);
    \end{tikzpicture}
    \caption{
        The dependency graph for \Cref{fig:example_rtlola}.
    }
    \label{fig:dep_graph_example}
    \end{minipage}
    \hfill
    \begin{minipage}{0.55\linewidth}
    \centering
    \includegraphics[width=\linewidth]{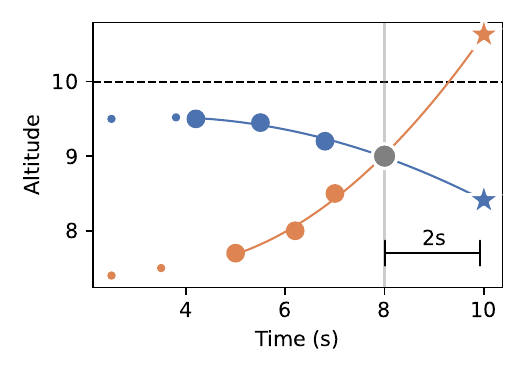}
    \vspace*{-7mm}
    \caption{An example of the prediction given two different altitude traces.}
    \label{fig:predict_example_plot}
    \end{minipage}
\end{figure}

The predict operator enables the monitor to predict future stream values based on past observations.
As an example, consider again the situation from \Cref{sec:preliminaries} monitoring the altitude of a drone.
By incorporating prediction, the monitor can issue early warnings when a violation is likely to occur in the near future, enabling countermeasures to be taken before the violation occurs.
\Cref{fig:predict_example_plot} illustrates this behavior for two altitude traces (orange and blue).
At each evaluation point, the prediction operator considers the last four observed values of the respective trace (highlighted in the figure) and fits a regression model to these points.
This model is used to estimate the altitude value two seconds into the future.
As shown, one trace yields a predicted violation while the other remains within safe bounds.
This can be expressed directly in \RTLola as follows:
\begin{lstlisting}
    output prediction @1s@ := alt_smoothed.predict(using: polyreg2,
                                over_discrete: 4, in: 2s).defaults(to: 0.0)
    trigger prediction > 10.0 "height violation expected soon"
\end{lstlisting}
Here, a second-degree polynomial is fitted to the four most recent altitude values, and the resulting model is evaluated at the specified future time point to obtain the predicted value.
If not all four values exist yet, a default value is used instead.

Formally, we introduce prediction as a pure operator over finite sequences of timestamped values.
The stream-based syntax presented above is then defined as syntactic sugar that extracts the required observations and applies the core operator.
This separation decouples the prediction mechanism from \RTLola's stream semantics and allows the existing static analyses to remain unchanged.

\subsection{Core Prediction Operator}

We first define a prediction function independently of \RTLola's syntax.
This definition is generic in the choice of a prediction model $\gamma$, requiring only that $\gamma$ provides a function construction via $\gamma.\mathit{fit}$.
\begin{definition}[Prediction Function]
    Let $\gamma$ be a prediction model and $n \in \NN_{>0}$.
    The prediction function $\gamma_n$ constructs a function based on a sequence of $n$ pairs of timestamps and values.
    \begin{align*}
        &\gamma_n : (\RR^+_\bot \times \VV_{\bot})^n \rightarrow (\RR^+ \rightarrow \VV_{\bot})\\
        &\gamma_n((t_1,v_1),\dots,(t_n,v_n)) = \begin{cases}
            f_{\bot}&\text{if } \exists 1\le i \le n. v_i = \bot\\
            \gamma.\mathit{fit}((t_1,v_1),\dots,(t_n,v_n)) &\text{otherwise}
        \end{cases}
    \end{align*}
    Here, $f_{\bot}: \RR^+ \rightarrow \bot$ denotes the constant function returning $\bot$ for all timestamps.
    Otherwise, $\gamma.\mathit{fit}((t_1,y_1),\dots,(t_n,v_n))$ yields a function mapping real-time timestamps to predicted values.
\end{definition}
Different instantiations of the prediction model $\gamma$ correspond to different regression or extrapolation techniques.
For instance, a linear regression model fits a line to the observed data points, while higher-order polynomial models capture nonlinear trends.
In practice, the choice of model can be adapted to the monitoring application and the computational overhead that can be tolerated at runtime.

\subsection{Syntactic Sugar}

We integrate prediction into \RTLola using a dedicated prediction expression.
\begin{definition}[Prediction Operator]\label{def:predict_syntax}
Let $s$ be a stream, $\gamma$ a prediction model, $n \in \mathbb{N}_{>0}$, and $\delta \in \RR^+$ a positive real-time offset.
The stream expression
\begin{lstlisting}[mathescape=true]
s.predict(using: $\gamma$, over_discrete: $n$, in: $\delta$)
\end{lstlisting}
evaluates at time $t$ to $\gamma_n\big(v_{s,n,t}\big)\big(\world(t) + \delta\big)$, where $v_{s,n,t}$ denotes the sequence of the $n$ most recent timestamped value pairs of $s$ available at $t$.
\end{definition}
Intuitively, the operator evaluates at each time $t$ by extracting the $n$ most recent timestamped values of $s$ available at $t$, constructing a prediction function via $\gamma_n$, and evaluating this function at the future timestamp $\world(t) + \delta$.
If any of the selected points is $\bot$, the result is $\bot$ as well.
For this case, the prediction must be followed by a default expression.

As syntactic sugar, we can make use of past \lstinline!offset! expressions to extract the required points.
For a stream $s$, we first collect the most recent $n$ values together with their corresponding timestamps using discrete offset accesses:
\begin{lstlisting}
    output $s$_points$_n$ :=
        (time($s$).offset(by: -($n-1$)), $s$.offset(by: -($n-1$))),
        ...
        (time($s$).offset(by: -1), $s$.offset(by: -1)),
        (time($s$), $s$)
\end{lstlisting}
Here, \lstinline!time(s)! denotes a stream that shares the same pacing as $s$ but carries the corresponding timestamps of its evaluation points.
Due to the synchronous nature of the offset access, the pacing of this derived stream inherits the pacing of $s$.
Consequently, at every evaluation point of $s$, the tuple of the most recent $n$ points is computed.
Note that \lstinline!offset! returns $\bot$ if the accessed value does not exist, which is then handled by the prediction function.
A prediction expression as given in \Cref{def:predict_syntax} is translated into the core operator application
\begin{lstlisting}
    $\gamma_n$($s$_points$_n$.hold())(time + $\delta$)
\end{lstlisting}
where \lstinline!time! gives access to the current timestamp of the stream evaluation.
The \lstinline!hold! access is required to decouple the pacing of $s$ from the evaluation of the prediction expression.
This allows the prediction operator to be used in streams with a different pacing than the predicted stream.

%% file: content/offsets.tex
\section{Future Offsets}\label{sec:future_offset}

\begin{figure}[t]
\begin{lstlisting}
output takeoff @1s@ := alt_smoothed $\le$ 0.05
    $\land$ alt_smoothed.offset(by: 1).defaults(to: 0.0) > 0.05
output valid_flight := alt_reached $\lor$ ($\neg$takeoff $\land$ valid_flight.offset(by: 1).defaults(to: false))
trigger $\neg$(takeoff $\Rightarrow$ valid_flight.offset(by: 1).defaults(to: false))
\end{lstlisting}
\caption{
    An \RTLola specification monitoring consecutive flights of a drone for altitude thresholds.
}
\label{fig:ltl_altitude_spec}
\end{figure}


The future offset operator enables references to stream values that will be available at a fixed position in the future relative to the current evaluation point.
To illustrate the operator, consider the LTL formula from the introduction:
\[
\square(\mathit{takeoff} \Rightarrow \bigcirc(\neg\mathit{takeoff} \;\mathcal{U}\; \mathit{alt\_reached})).
\]
Intuitively, the monitor ensures that after each takeoff, the drone eventually reaches a valid flight altitude before the next takeoff occurs.
A corresponding specification can be expressed in \RTLola with future offsets as depicted in \Cref{fig:ltl_altitude_spec}, closely following the structure of the LTL formula.
The stream \lstinline!takeoff! detects takeoff events based on an altitude threshold and evaluates to true when the altitude is initially below the bound and exceeds it in the next time step.
The stream \lstinline!valid_flight! captures whether the current flight (until the next takeoff event) satisfies the altitude requirement.
The trigger then ensures that every takeoff event corresponds to a valid flight.

\begin{figure}[t]
    \centering
    \includegraphics[width=0.55\linewidth]{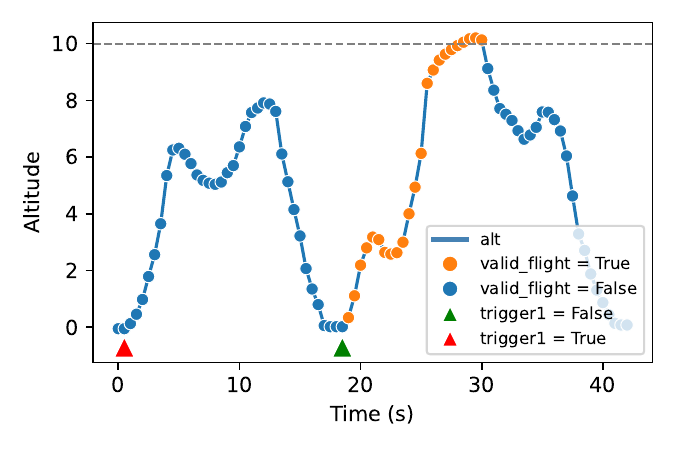}
    \caption{The altitude monitoring of two consequtive flight of a drone.}
    \label{fig:future_offset_example}
\end{figure}

An example is illustrated in \Cref{fig:future_offset_example}.
The plot shows an altitude trace of a drone performing two consecutive flights.
The color of the plotted points encodes the value of \lstinline!valid_flight!.
The triangles below the plot mark the takeoff events detected by the monitor.
Their color corresponds to the evaluation of the trigger condition:
On the first flight, the target altitude is not reached, so the trigger issues a warning.
On the second flight, the altitude is reached successfully.
Note that the violation is reported with the takeoff timestamp, allowing it to be associated with the correct flight that caused it.

In the remainder of this section, we first introduce real-time offsets in \RTLola, which naturally arise in future-oriented specifications, and show how they can be reduced to discrete offsets.
We then formalize the semantics and adapt the static analyses to the extended offset operator.

\subsection{Real-Time Offsets}

In contrast to the prediction operator, offset accesses refer only to concrete stream values.
The offset operator therefore operates over discrete positions in the stream, accessing the $n$-th value in the past or in the future.
However, many specifications are more naturally expressed using real-time offsets, such as "the value two seconds ago" or "the value in five seconds".

In general, such references are not directly supported, since the monitor cannot guarantee that a value exists at an exact real-time position during execution.
However, if the pacing information of streams provides sufficient knowledge about their temporal structure, a real-time offset can be translated into an equivalent discrete offset.
We therefore introduce real-time offsets not as an operator with a new semantics, but rather as syntactic sugar for discrete offsets.

\begin{definition}[Rewriting Real-Time Offsets]
  Let $s$ be a stream with periodic pacing $\delta_s$, and let $t\in\RR$ be a real-time duration.
  If there exists $k \in \NN$ such that $ t = k \cdot \delta_s$, then the following rewrite rule applies:
  \[
	s.\rtoffset{t} = s.\offset{k}.
  \]
  If no such $k$ exists, the monitor cannot guarantee that the accessed value exists.
  Such expressions are rejected during static analysis.
\end{definition}

\subsection{Future-Offset Semantics}

We refer to \RTLola extended with the future offset operator as \FRTLola, and extend \RTLola's semantics from~\cite{baumeister2025intermediate} with two inference rules for future offsets:
\[
\inferrule*[left=\text{FO-Hit}]
  {o > 0 \\ \mathit{sfx} = \mathit{Sfx}(\world, s, t) \\\\ |\mathit{sfx}| > o \\ \mathit{sfx}[o] = v}
  {s.\offset{o} \Downarrow_\world^t v}
\quad
\inferrule*[left=\text{FO-Miss}]
  {\mathit{o} > 0 \\\\ \mathit{sfx} = \mathit{Sfx}(\world, s, t) \\ |\mathit{sfx}| \le o}
  {s.\offset{o} \Downarrow_\world^t \bot}
\]
These rules are analogous to their past-offset counterparts (where $o < 0$), but operate on the suffix of the trace rather than its prefix.
Specifically, $\mathit{Sfx}(\world, s, t)$ denotes the sequence of all values of stream $s$ at time $t$ and beyond, retaining only positions where $\world(s)(t) \neq \bot$.
\textsc{FO-Hit} applies when the suffix is long enough to contain index $o$, yielding the value $v$ found there.
\textsc{FO-Miss} applies when the suffix is too short and returns $\bot$, which must be resolved by a default expression.
 
\subsection{Well-Formedness}

%


The addition of future offsets to \FRTLola introduces a new source of cyclic dependencies that must be considered for well-formedness.

\begin{definition}[Well-Formedness]\label{def:well-formedness}
A \FRTLola specification $\varphi$ is called well-formed, if the dependency graph $G_\varphi$ does not contain any cycle with:
\begin{enumerate}
    \item accumulated edge weight of zero.
    \item a \lstinline!Filter!-edge and an \lstinline!Offset!-edge.
    \item an positive \lstinline!Offset!-edge and either an \lstinline!Aggregate!-edge or a \lstinline!Hold!-edge.
\end{enumerate}
\end{definition}

The first criterion originates from \Lola~\cite{DBLP:conf/time/DAngeloSSRFSMM05} and excludes synchronous cyclic dependencies.
The second criterion was introduced for \RTLola~\cite{DBLP:phd/dnb/Schwenger22} to account for the interaction between filtering and offset accesses.
The third criterion is specific to \FRTLola and addresses asynchronous cycles involving future offsets.
It is best illustrated with an example.
We use \lstinline!s.offset(by: o, or: d)! as syntactic sugar for the expression \lstinline!s.offset(by: o).defaults(to: d)!:

\begin{center}
\begin{minipage}{0.62\linewidth}
\begin{lstlisting}
input x : Bool
input y : Bool
output a @x$\textcolor{greentypes}{\land}$y@ := b.offset(by: 2, or: false)
output b @y@ := a.hold(or: false)
\end{lstlisting}
\end{minipage}
\hfill
\begin{minipage}{0.37\linewidth}
\hfill
\begin{tikzpicture}[
    n/.style={draw,circle,fill=black,inner sep=0.6mm}
]
    \def\f{0.25}
    \def\l{3.3}
    \def\t{0.7}
    \draw[-{Stealth}] (0,0) node[left] {\lstinline!x!} -- ++(\l,0);
    \draw[-{Stealth}] (0,-1*\f) node[left] {\lstinline!y!} -- ++(\l,0);
    \draw[-{Stealth}] (0,-2*\f) node[left] {\lstinline!a!} -- ++(\l,0);
    \draw[-{Stealth}] (0,-5*\f) node[left] {\lstinline!b!} -- ++(\l,0);

    \foreach \i in {0,3}{
        \node[n] (x\i) at (\t*\i+0.5,0) {};
        \node[n] (a\i) at (\t*\i+0.5,-2*\f) {};
    }
    \foreach \i in {0,1,2,3}{
        \node[n] (y\i) at (\t*\i+0.5,-1*\f) {};
        \node[n] (b\i) at (\t*\i+0.5,-5*\f) {};
    }

    \draw[->,shorten >=1mm] (a0) to[bend left=15] (b2);
    \draw[->,dashed,shorten >=1mm] (b2) to[bend left=15] (a0);
\end{tikzpicture}
\end{minipage}
\end{center}

On the right, we depict one possible \emph{evaluation graph} of the specification.
The evaluation graph represents the concrete dependencies between stream instances at particular evaluation points.
Since the timing of input streams \lstinline!x! and \lstinline!y! is not in control of the monitor, it cannot statically predict when \lstinline!a! and \lstinline!b! will be evaluated.
In the example on the right, you can see one possible evaluation.
However, in this case, the dashed \lstinline!hold! access from \lstinline!b! to \lstinline!a! forms an evaluation cycle.
The specification is therefore not well-defined, even though the corresponding cycle in the dependency graph contains a positive weight.
The third well-formedness criterion excludes precisely such situations.

The well-formedness conditions exclude dependency patterns that may lead to unresolved cyclic dependencies during evaluation.
This allows us to establish the central correctness result of the analysis.
\begin{theorem}\label{thm:WF_WD}
Every well-formed specification is well-defined.
\end{theorem}
\ifthenelse{\boolean{fullversion}}{
The proof is given in \Cref{app:wd_proof}.
}{
The proof is given in the full version~\cite{fullversion} of this paper.
}

\paragraph{Periodic Dependency Cycles.}
For cycles that consist solely of periodic streams, the analysis can be less restrictive.
Consider the following example:
\begin{center}
\begin{minipage}{0.6\linewidth}
\begin{lstlisting}
output a @3s@ := b.offset(by: 2, or: false)
output b @3s@ := c.offset(by: -2, or: false)
output c @1s@ := a.hold(or: false)
\end{lstlisting}
\end{minipage}
\hfill
\begin{minipage}{0.39\linewidth}
\hfill
\begin{tikzpicture}[
    n/.style={draw,circle,fill=black,inner sep=0.6mm}
]
    \def\f{0.5}
    \def\l{3.3}
    \def\t{0.4}
    \draw[-{Stealth}] (0,0) node[left] {\lstinline!a!} -- ++(\l,0);
    \draw[-{Stealth}] (0,-1*\f) node[left] {\lstinline!b!} -- ++(\l,0);
    \draw[-{Stealth}] (0,-2*\f) node[left] {\lstinline!c!} -- ++(\l,0);

    \foreach \i in {0,3,6}{
        \node[n] (a\i) at (\t*\i+0.3,0) {};
        \node[n] (b\i) at (\t*\i+0.3,-1*\f) {};
    }
    \foreach \i in {0,...,6}{
        \node[n] (c\i) at (\t*\i+0.3,-2*\f) {};
    }

    \draw[->,shorten >=1.5mm] (a0) -- ([yshift=0.5mm]b6.center);
    \draw[->,shorten >=1mm] (b6) -- (c4);
    \draw[->,shorten >=1mm,dashed] (c4) to[bend right] (a3);
\end{tikzpicture}
\end{minipage}
\end{center}

Although the specification includes a cycle with an accumulated edge weight of zero, and therefore does not satisfy well-formedness according to \Cref{def:well-formedness}, it does not induce a cycle in the evaluation graph and is well-defined.
At every evaluation point, all stream values can be computed without cyclic dependencies.
For such specifications, we provide a transformation that uses the periodic pacing of the streams in the cycle.
It assigns all streams a common pacing corresponding to the greatest common divisor of their original periods and adjusts the offset accesses accordingly.
After this step, the resulting dependency cycle becomes synchronous and no longer forms a zero-cycle.
\ifthenelse{\boolean{fullversion}}{
You can find more details about this transformation in \Cref{app:period_cycle}.
}{
You can find more details about this transformation in the full version~\cite{fullversion} of this paper.
}

\subsection{Monitoring Algorithm}

\begin{algorithm}[t]
\caption{Asynchronous Monitoring Algorithm}
\label{algo:monitor}
\DontPrintSemicolon

\KwIn{specification $\varphi$,
    input stream assignment $I : \Time \to (\Iref \to \VV_\bot)$,
    time map $T: \Time \rightarrow \RR^+$
}

\BlankLine
$R \leftarrow \emptyset \subseteq \powerset{(\Iref \uplus \Oref) \times \VV_\bot \times \Time}$ \tcc*{Resolved values} 
$U \leftarrow \emptyset \subseteq \powerset{\Oref \times \Time}$ \tcc *{Unresolved outputs} 
$M \leftarrow \emptyset \subseteq \powerset{\Oref \times \Time}$ \tcc*{Maybe outputs}

\BlankLine
\For{$t_M \in 1, 2, \ldots, T_{\mathit{max}}$}{
    $R \leftarrow R \;\cup\; \bigl\{\,(s, I(t_M)(s), t) \mid s \in \Iref,\; I(t_M)(s) \neq \bot\,\bigr\}$\;
    $M \leftarrow M \;\cup\; \bigl\{\,(s, t_M) \mid s \in \Oref\,\bigr\}$\;

    \BlankLine
    \While{\upshape some set in $\{R, U, M\}$ has changed}{

        \ForEach{$(s,t) \in M$}{
            $M \leftarrow M \setminus \{(s,t)\}$\;
            \Switch{$\mathit{eval}(\varphi^p_s, R, U, M, t_M, t, T)$}{
                \lCase{$?$}{
                    $M \leftarrow M \cup \{(s, t)\}$
                }
                \lCase{$\mathit{true}$}{
                    $U \leftarrow U \cup \{(s, t)\}$
                }
                \lCase{$\mathit{false}$}{
                    skip
                }
            }
        }

        \BlankLine
        
        \ForEach{$(s,t) \in U$}{
            $U \leftarrow U \setminus \{(s,t)\}$\;
            \Switch{$\mathit{eval}(\varphi_s, R, U, M, t_M, t, T)$}{
                \lCase{$?$}{
                    $U \leftarrow U \cup \{(s,t)\}$
                }
                \lCase{$v$}{
                    $R \leftarrow R \cup \{(s, t)\}$
                }
            }
        }
    }
}
\end{algorithm}

\Cref{algo:monitor} presents the monitoring algorithm for \FRTLola.
The algorithm maintains three sets.
$R$ holds the resolved stream values that have been computed.
$U$ tracks timepoints at which stream values will definitely be computed, but for which the stream expression cannot be evaluated yet because of missing dependencies.
Finally, $M$ contains timepoints that \emph{might} be evaluated, depending on whether the pacing condition is satisfied at that timepoint.

The algorithm is similar to that of \Lola~\cite{DBLP:conf/time/DAngeloSSRFSMM05}, with the addition of the set $M$.
The fixpoint loop alternates between pacing and value evaluation using the function $\mathit{eval}$.
In each iteration, every pair currently contained in $M$ (respectively $U$) is evaluated exactly once in arbitrary order, based on a snapshot of the set taken at the start of that iteration.
For pacing, $\mathit{eval}$ is applied to $\varphi^p_s$, the combined pacing and when-expression of stream $s$, and returns a boolean.
For value computation, it is applied to $\varphi_s$, the with-expression of $s$.
In both cases, $\mathit{eval}$ may return $?$ (unresolved) when not all required dependencies are available yet.

The handling of offset accesses in $\mathit{eval}$ deserves particular attention.
To evaluate an offset expression, the correct accessed value needs to be found in $R$.
The algorithm resolves this by searching through $R$ and $U$ forwards or backward in time, until the desired offset is found.
If found in $R$, the corresponding value is returned.
However, if an entry in $M$ is encountered along the way, the evaluation returns unresolved because it can not be decided which value will be accessed.

A further difference from \Lola is that this algorithm does not partially evaluate expressions.
Instead, all values accessed by an expression are retained in memory until the expression can be evaluated in full.
This has the advantage that if a stream value is required by multiple expressions, it needs to be stored only once, rather than duplicated across partial evaluation states.

To relate \Cref{algo:monitor} to the semantics of \FRTLola, we establish a correspondence between the values in $R$ and in the evaluation model.

\begin{theorem}[Correctness of Monitoring Algorithm]\label{thm:algorithm}
Given a well-defined specification $\varphi$, an input stream assignment $I$, and a time map $T$.
Let $\world \in \semantics{\varphi}$ be the unique evaluation model for the given input streams and time map, and let $R_i$ denote the set $R$ after the completion of iteration $i$ of the algorithm.
Then, for every $i \in \Time$, the following holds:
\begin{enumerate}
    \item If $(s, v, t) \in R_i$, then $\world(s)(t) = v$.
    \item If $\world(s)(t) = v$, then there exists $i' \ge i$ such that $(s,v,t) \in R_{i'}$.
\end{enumerate}
\end{theorem}

The theorem establishes that the monitoring algorithm is both sound and complete with respect to the semantics of \FRTLola.
Soundness ensures that every value reported by the monitor algorithm corresponds to the value in the evaluation model, while completeness guarantees that all values in the evaluation model are eventually produced by the algorithm.
\ifthenelse{\boolean{fullversion}}{
The proof can be found in \Cref{app:algo_proof}.
}{
The proof can be found in the full version~\cite{fullversion} of this paper.
}

\subsection{Memory Analysis}

The memory bound of a stream specifies how many past values must be retained to ensure the correct evaluation of all stream expressions in a specification.
This bound may be unbounded, as is the case for the specification in \Cref{fig:ltl_altitude_spec}, because we can not statically know how long between two takeoff events.
A specification with a finite memory bound is called \emph{efficiently monitorable}, a term originally introduced by \Lola~\cite{DBLP:conf/time/DAngeloSSRFSMM05}.
In contrast to \Lola, our algorithm does not partially evaluate unresolved expressions.
While \Lola can discard intermediate results once they are inlined in the expression, we retain all accessed values until a stream expression can be fully evaluated.

\begin{figure}[t]
\begin{center}
\begin{minipage}{0.65\linewidth}
\begin{lstlisting}
input alt : Float64
output n_alt := alt.offset(by:1, or:0.0)
output slope := n_alt.offset(by:-1, or:0.0)
                  - n_alt.offset(by:1, or:0.0)
\end{lstlisting}
\end{minipage}
\hfill
\begin{minipage}{0.34\linewidth}
\hfill
\hspace*{-3mm}
\begin{tikzpicture}[
    n/.style={draw,circle,fill=black,inner sep=0,minimum width=4mm},
    ur/.style={n,red,fill=white,execute at begin node={\scriptsize\strut ?}},
    fg/.style={n,dotted,gray,fill=white,execute at begin node={\scriptsize\strut $\checkmark$}},
    rs/.style={n,green!50!black,fill=white,execute at begin node={\scriptsize\strut $\checkmark$}},
    shorten >=0.5mm
]
    \def\f{0.6}
    \def\l{3.3}
    \def\t{0.7}
    \draw[-{Stealth}] (0,0) node[left,font=\scriptsize] {\lstinline!alt!} -- ++(\l,0);
    \draw[-{Stealth}] (0,-1*\f) node[left,font=\scriptsize] {\lstinline!n_alt!} -- ++(\l,0);
    \draw[-{Stealth}] (0,-2*\f) node[left,font=\scriptsize] {\lstinline!slope!} -- ++(\l,0);

    \foreach \i/\s in {0/fg,1/fg,2/fg,3/rs}{
        \node[\s] (a\i) at (\t*\i+0.5,0) {};
    }
    \foreach \i/\s in {0/rs,1/rs,2/rs,3/ur}{
        \node[\s] (b\i) at (\t*\i+0.5,-1*\f) {};
    }
    \foreach \i/\s in {0/fg,1/rs,2/ur,3/ur}{
        \node[\s] (c\i) at (\t*\i+0.5,-2*\f) {};
    }

    \draw[->] (c1) -- (b0);
    \draw[->] (c1) -- (b2);
    \draw[->] (b2) -- (a3);
    
\end{tikzpicture}
\end{minipage}
\end{center}
\caption{Example specification with corresponding evaluation graph.}
\label{fig:unresolved_resolved_example}
\end{figure}
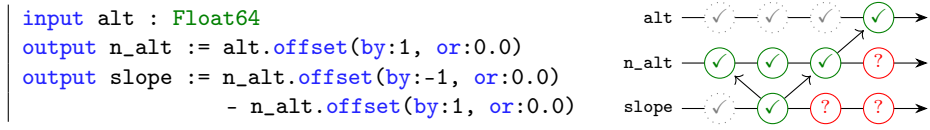

We illustrate this with the example in \Cref{fig:unresolved_resolved_example}.
The figure on the right illustrates an evaluation graph where unresolved values are shown in red, and resolved values are marked with a checkmark.
Green nodes indicate values that are still required for future evaluations, whereas gray nodes are no longer needed and can be discarded.
From this graph, the memory bound can be read directly:
for \lstinline!alt!, the bound is 1 since only the current value must be retained.
For \lstinline!n_alt! the bound is 4, as the fourth last value is still required to evaluate \lstinline!slope!.
This is because the execution of \lstinline!slope! is delayed until all dependencies are resolved.

This delay is captured by the \emph{delay} measure~\cite{DBLP:conf/rv/FinkbeinerOPS20}, which assigns to each stream the maximal temporal displacement induced by its dependencies.
Intuitively, the delay $\triangle(s)$ of a stream $s$ measures how far into the future the monitor must look to compute a new value for $s$.
It equals the largest total offset accumulated along any path originating at $s$ in the dependency graph.
The delay becomes infinite when any such path passes through an asynchronous access to a stream that itself has a non-zero delay, since the waiting time is then unbounded.

Based on this notion, we define the memory bound of a stream.
\begin{definition}[Memory Bound]
    Let $\varphi$ be a well-defined \FRTLola specification with the corresponding dependency graph $G_\varphi=\langle V, E \rangle$.
    For a stream $s$, the \emph{memory bound} $\mu(s)$ is defined as:
    \[
        \mu_\varphi(s) = \max \left\{ \triangle(s), \max \left\{ \triangle(s') - k \mid (s',(\_ , k),s) \in E \right\} \right\} + 1.
    \]
\end{definition}

The following theorem shows that this bound is sufficient for the monitoring algorithm and therefore characterizes the amount of memory that must be retained for each stream.
\ifthenelse{\boolean{fullversion}}{
The proof can be found in \Cref{app:proof_mem_bound}.
}{
The proof can be found in the full version~\cite{fullversion} of this paper.
}
\begin{theorem}[Memory Bound Soundness]\label{thm:memory_bound}
Let $\varphi$ be a well-defined \FRTLola specification and let $\mu_\varphi$ be its memory-bounds.
During the execution of \Cref{algo:monitor}, for every stream $s$ and every value $(s,v,t)\in R$, if there exist $\mu(s)$ newer values of $s$ in $R$, then $(s,v,t)$ can be removed from $R$ without affecting the correctness of the algorithm.
\end{theorem}

\begin{figure}[t]
\begin{center}
\begin{minipage}{0.65\linewidth}
\begin{lstlisting}
input alt : Float64
output n_alt := alt.offset(by: 1, or:0.0)
output p_alt @1s@ := n_alt.hold(or: 0.0)
\end{lstlisting}
\end{minipage}
\hfill
\begin{minipage}{0.34\linewidth}
\hfill
\hspace*{-3mm}
\begin{tikzpicture}[
    n/.style={draw,circle,fill=black,inner sep=0,minimum width=4mm},
    ur/.style={n,red,fill=white,execute at begin node={\scriptsize\strut ?}},
    fg/.style={n,dotted,gray,fill=white,execute at begin node={\scriptsize\strut $\checkmark$}},
    rs/.style={n,green!50!black,fill=white,execute at begin node={\scriptsize\strut $\checkmark$}},
    shorten >=0.5mm
]
    \def\f{0.55}
    \def\l{3.3}
    \def\t{0.24}
    \draw[-{Stealth}] (0,0) node[left,font=\scriptsize] {\lstinline!alt!} -- ++(\l,0);
    \draw[-{Stealth}] (0,-1*\f) node[left,font=\scriptsize] {\lstinline!n_alt!} -- ++(\l,0);
    \draw[-{Stealth}] (0,-2*\f) node[left,font=\scriptsize] {\lstinline!p_alt!} -- ++(\l,0);

    \node[rs] (a0) at (\t*0+0.5,0) {};
    \node[ur] (b0) at (\t*0+0.5,-1*\f) {};
    \draw[->]  (b0) -- ++(2cm,4mm) node[right,fill=white] {?};

    \foreach \i in {0,...,3}{
        \node[ur] (c\i) at (0.7*\i+0.5,-2*\f) {};
    }
    \begin{scope}[on background layer]
    \draw[dashed,->,shorten >=0.3mm] (c0.north) -- (b0);
    \draw[dashed,->,shorten >=2mm] (c1.north) -- (b0);
    \draw[dashed,->,shorten >=3mm] (c2.north) -- (b0);
    \draw[dashed,->,shorten >=4mm] (c3.north) -- (b0);
    \end{scope}
\end{tikzpicture}
\end{minipage}
\end{center}
\caption{Example for an unbounded specification without dependency cycles.}
\label{fig:unbounded_hold_example}
\end{figure}
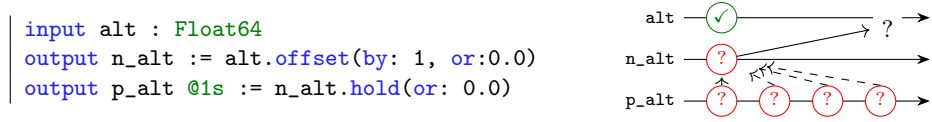

Note that, in contrast to \Lola, unbounded memory requirements do not necessarily arise from dependency cycles.
They can also occur due to unresolved asynchronous accesses.
Consider the specification in \Cref{fig:unbounded_hold_example}:
The stream \lstinline!n_alt! depends on the next value of \lstinline!alt! and remains unresolved until that value arrives.
The periodic stream \lstinline!p_alt! accesses \lstinline!n_alt! asynchronously, and therefore cannot be resolved until the next altitude reading arrives.
Since the arrival time of the next input value is not known statically, the monitor cannot determine how many evaluations of \lstinline!p_alt! will occur before \lstinline!n_alt! becomes available.

\subsection{Evaluation Order}

The monitoring algorithm in \Cref{algo:monitor} operates as a fixpoint computation, iterating until no more values can be resolved.
In practice, a well-chosen order can avoid unnecessary iterations.
An evaluation order already existed for \RTLola~\cite{DBLP:phd/dnb/Schwenger22} where the dependency graph induces a topological order on streams.
However, the original analysis considered only zero-offset edges, since past offset values are already computed in previous iterations.
Future offsets break this assumption: a stream with a future offset access can only be evaluated once the referenced future value has been resolved.
By accounting for the delay $\triangle$ of each stream, the dependency graph can be adjusted so that future offset accesses resemble synchronous zero-offset dependencies, from which a valid stream evaluation order is obtained by topological sort.
Within each iteration, timepoints are processed in descending order: once a future value is resolved, all earlier stream instances that depend on it are resolved in the same pass, drastically reducing the number of fixpoint iterations.
For the specification in \Cref{fig:unresolved_resolved_example} this yields the stream order $\texttt{alt} \prec \texttt{n\_alt} \prec \texttt{slope}$.

%% file: content/evaluation.tex
\section{Evaluation}\label{sec:evaluation}

In this section, we evaluate our implementation based on the \RTLola StreamIR Interpreter~\cite{DBLP:conf/cav/BaumeisterCFS25}.
\ifthenelse{\boolean{fullversion}}{
All results were obtained on a system with a 13th Gen Intel Core i7-1355U, and the corresponding specifications are provided in \Cref{app:specs}.
}{
All results were obtained on a system with a 13th Gen Intel Core i7-1355U, and the corresponding specifications are provided in the full version~\cite{fullversion} of this paper.
}

\begin{figure}[t]
    \begin{subfigure}{0.495\linewidth}
    \includegraphics[width=\linewidth]{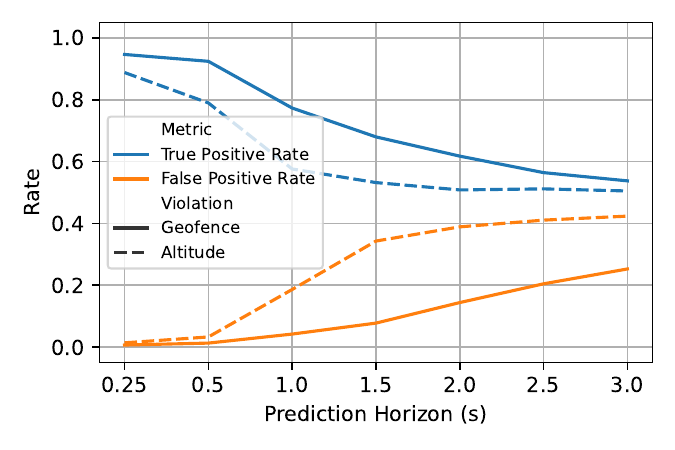}
    \caption{True and false positive rates over prediction horizon.}
    \label{fig:prediction_rates}
    \end{subfigure}
    \hfill
    \begin{subfigure}{0.495\linewidth}
    \includegraphics[width=\linewidth]{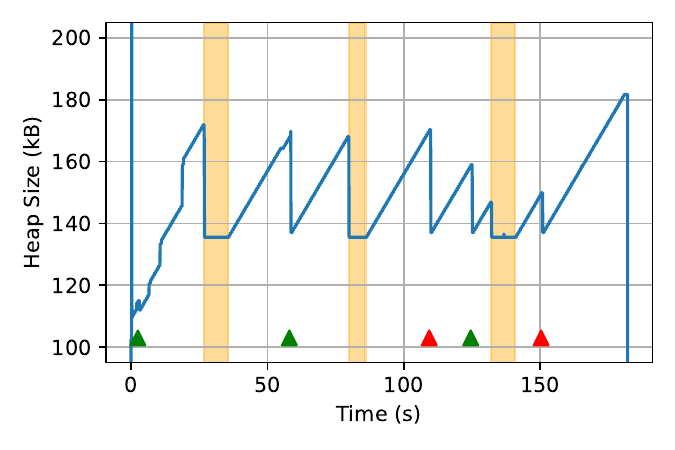}
    \caption{Memory usage for unbounded specification.}
    \label{fig:memory_usage}
    \end{subfigure}
\end{figure}

\paragraph{Prediction.}
To evaluate the prediction operator, we make use of a specification from \cite{DBLP:conf/rv/BaumeisterFS25}, which involved monitoring the altitude and geofence bounds of a drone.
Here, we evaluate the true positive rate (TPR) and false positive rate (FPR) when predicting trigger violations for different prediction horizons.
For each prediction horizon, we obtain a predicted trigger verdict by extrapolating the monitored value across the horizon and checking it against the threshold, and compare it to the actual trigger verdict once that time is reached.
All results use a fixed observation window of 5 values and a linear prediction method, evaluated across the 10 drone flight traces from \cite{DBLP:conf/rv/BaumeisterFS25} recorded in the AirSim simulator.

\Cref{fig:prediction_rates} shows the TPR and FPR for both violation types across prediction horizons ranging from 0.5 to 3.0 seconds.
At 0.25 seconds, the geofence monitor achieves a TPR of $94.7\%$ at an FPR of $0.7\%$, while the altitude monitor yields a TPR of $88\%$ at an FPR of $1.4\%$.
As expected, performance degrades with increasing horizon, with the altitude being affected more strongly than the geofence.
These results demonstrate that the prediction operator can detect violations more than one second in advance with a high accuracy, which is significant for cyber-physical systems operating on a millisecond timescale.

\paragraph{Bounded Specifications.}
To evaluate bounded specifications, we monitor a geofence recovery requirement: whenever the system exits the geofence, it must correct itself and re-enter the geofence within 2s.
This obligation is most naturally expressed using a future real-time offset, which directly mirrors the forward-looking deadline in the requirement:
\begin{lstlisting}
output @1Hz@ outside_fence : Bool := ...
trigger $\neg$(outside_fence $\Rightarrow$ $\neg$outside_fence.offset(by: 2s, or: false))
\end{lstlisting}
The same property can equivalently be expressed using only past-time offsets, at the cost of inverting the reading:
\begin{lstlisting}
trigger $\neg$(outside_fence.offset(by: -2s, or: false) $\Rightarrow$ $\neg$outside_fence)
\end{lstlisting}

\begin{wrapfigure}[12]{r}{0.38\textwidth}
  \vspace*{-4mm}
  \centering
  \includegraphics[width=\linewidth]{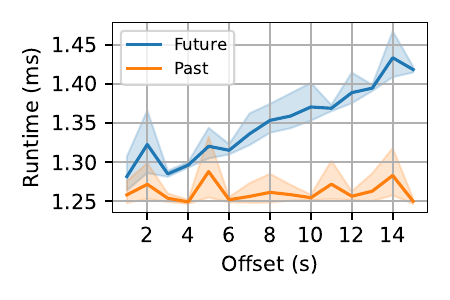}
  \caption{Runtime comparison.}
  \label{fig:bounded_runtime}
\end{wrapfigure}
We evaluate the runtime overhead and memory consumption of both formulations by monitoring all traces consecutively across 200 runs.
\Cref{fig:bounded_runtime} compares the runtime of both specifications for different offset values.
While the runtime of the past-time specification remains essentially constant as the offset increases, the runtime of the future-time specification grows steadily.
This is because unresolved values must be revisited until they can be resolved, causing the monitoring algorithm to process an increasing number of pending values.
Nevertheless, even at the largest evaluated offset, the future-time specification incurs only a $13.5\%$ runtime overhead compared to its past-time counterpart, and without the cost of inverting the timepoints of the triggers.
In contrast, memory consumption is nearly identical for both specifications and increases only marginally with larger offsets.

\paragraph{Unbounded Specifications.}
Finally, we evaluate the memory usage of monitors with unbounded memory.
\Cref{fig:memory_usage} shows the memory usage for three drone flight traces from \cite{DBLP:conf/rv/BaumeisterFS25} while monitoring the specification from \Cref{fig:ltl_altitude_spec}.
During monitoring, triangles at the bottom of the plot mark takeoff events, and orange-shaded regions indicate intervals during which the target altitude is reached.
As visible in the figure, the memory resets whenever either a takeoff or an altitude reached event occurs.
Since takeoff events occur at roughly known intervals, the memory consumption is effectively bounded in practice.

%% file: content/conclusion.tex
\section{Conclusion}


We have extended \RTLola with two complementary mechanisms for reasoning about future stream values: a prediction operator and a future-offset operator.
The prediction operator applies forecasting models to a sliding window of past observations, enabling lightweight, immediate estimates of future values.
The future-offset operator provides exact access to future stream values at the cost of a delayed evaluation.

These two mechanisms are particularly natural in combination: predictions can be issued immediately, while future-offset accesses later confirm or refute them against the ground truth.
Moreover, since predictions are inherently future-directed, it is natural to express the confirming specification in the same future-oriented style.
This improves the expressiveness and readability of specifications, while introducing only minor computational overhead.

As future work, we plan to investigate the automatic rewriting of bounded future-offset specifications into equivalent past-time specifications.
This would allow users to benefit from future-oriented specifications while avoiding their runtime overhead, by adapting evaluation timestamps such that verdicts are reported at the correct points in time.

%% file: content/appendix.tex
\section{Evaluation Graph}

\begin{definition}[Evaluation Graph]
Given a specification $\varphi$ and an input trace $I$, the evaluation graph
$\EG{I}$ is the least pair $(V,E)$ satisfying:

\begin{itemize}
    \item $V \subseteq ID^\uparrow \times Time$, where
    $(sid,t) \in V$ iff
    \begin{itemize}
        \item $t$ satisfies the pacing constraint of $sid$, and
        \item the \lstinline!when!-condition of $sid$ at time $t$
        evaluates to $\top$ using only values of stream instances
$(sid',t') \in V$.
    \end{itemize}

    \item $E \subseteq V \times V$, where
    $((sid,t),(sid',t')) \in E$ iff
    $(sid',t') \in V$ and the evaluation of $sid'$ at time $t'$
    accesses the value of $sid$ at time $t$
    (according to the \RTLola expression semantics).
\end{itemize}

If no such pair $(V,E)$ exists, then $\EG{I}$ is undefined.
\end{definition}

\begin{lemma}\label{lemma:evaluation_graph}
    If for every input trace $i$ the evaluation graph $\EG{I}$  exists and is acyclic, then $\varphi$ is well-defined.
\end{lemma}

\begin{proof}
    Let $I$ be an arbitrary input trace.
    We show that there exists a unique evaluation model $\world$ satisfying the stream-equations of $\varphi$.
    Since $\EG{I}$ exists, the set of evaluated stream instances is well-defined.
    Note that since $\EG{I}$ is acyclic, it admits a topological ordering $\prec$ on $V$.
    \begin{itemize}
        \item \textbf{Existence.} 
        We define $\world$ recursively along $\prec$:
        \begin{itemize}
            \item \textbf{Base Case:} For $(\sid, t)$ with $\sid \in \Iref$, set $\world(\sid)(t) = I(\sid)(t)$.
            \item \textbf{Recursive Case:} For $(\sid, t)$ with $\sid \in \Oref$, all dependencies $(\sid', t')$ of $(\sid, t)$ satisfy $(\sid',t')\prec (\sid,t)$, so $\world(\sid')(t')$ is already defined.
            We set $\world(\sid)(t)$ by evaluating the expression $\sid$ at $t$ using the already evaluated values.
        \end{itemize}
        \item \textbf{Uniqueness.} Suppose $\world,\world'$ are two evaluation models for $I$.
        We show $\world = \world'$ by induction along $\prec$.
        \begin{itemize}
            \item \textbf{Base Case:} For $\sid \in \Iref$, both models must satisfy $\world(\sid)(t) = I(\sid)(t)$.
            \item \textbf{Induction Step:} For $(\sid,t)$ with $\sid \in \Oref$, all dependencies $(\sid',t') \prec (sid,t)$ satisfy $\world(\sid')(t') = \world'(\sid')(t')$ by induction hypothesis.
            Since the expression of $\sid$ is a pure function of its dependencies, $\world(\sid)(t) = \world'(\sid)(t)$.
        \end{itemize}
    \end{itemize}
\end{proof}

\section{Proof for \Cref{thm:WF_WD}}\label{app:wd_proof}
\begin{proof}
    We prove the contraposition: If $\varphi$ is not well-defined, then $\varphi$ is not well-formed.
    
    Assume $\varphi$ is not well-defined.
    By the contraposition of \Cref{lemma:evaluation_graph}, there exists an input trace $I$ such that the evaluation graph $\EG{I}$ either does not exist or contains a cycle.
    \begin{itemize}
        \item \textbf{$\EG{I}$ does not exist.}
        The evaluation graph is undefined if it is not possible to construct the vertex set $V$ satisfying all pacing and \lstinline!when!-conditions.
        If the dependency relation induced by these conditions were acyclic, then $V$ could be constructed in a topological order.
        Therefore, the non-existence of $V$ implies that the dependency relation is cyclic.
        Since only \lstinline!when!-conditions determine whether a vertex
$(sid,t)$ exists, any such cycle must contain a Filter-edge.
        This contradicts the well-formedness criterion (2).
containing Filter-edges.
        \item \textbf{$\EG{I}$ contains a cycle.}
        Let
        \[
            C = (\sid_1, t_1) \rightarrow (\sid_2, t_2) \rightarrow \ldots \rightarrow (\sid_n, t_n) \rightarrow (\sid_1, t_1)
        \]
        be a cycle in $\EG{I}$.
        We define $\Delta_{C,i} = t_{i+1} - t_i$ for edge $i$ in $C$.
        Since $C$ is a cycle:
        \[
            \sum_{i=1}^n \Delta_{C,i} = 0
        \]
        Note that if edge $i$ is a
        \begin{itemize}
            \item discrete offset $o$: $\Delta_{C,i} = o \in \ZZ$,
            \item hold access: $\Delta_{C,i} \le 0$, since it accesses the latest value at or before the current time,
            \item aggregate: $\Delta_{C,i} \le 0$, since it accesses values in a past window.
        \end{itemize}
        Further, note that we count a synchronous access as one with an offset of 0.

        We distinguish two cases:
        \begin{itemize}
            \item \textbf{Every Hold and Aggregate Edge in $C$ has $\Delta_{C,i} = 0$.}
            \[
                \sum_{i=1}^n \Delta_{C,i} = \sum_{\text{edge $i$ is offset}} o_i + \sum_{\text{edge $i$ is hold/agg}} 0 = 0
            \]
            In this case, $C$ induces a cycle in the dependency graph with accumulated weight 0.
            This violates the well-formedness condition (1), so $\varphi$ is not well-formed.
            \item \textbf{There exists Hold or Aggregate in $C$ with $\Delta_{C,i} < 0$}
            Since all Hold and Aggregate edges satisfy $\Delta_{C,i} \le 0$, and at least one is strictly negative, their total contribution is negative:
            \[
                \sum_{\text{edge $i$ is hold/agg}} \Delta_{C,i} < 0.
            \]
            Since the total sum of $C$ is zero, the offset edges must compensate:
            \[
                \sum_{\text{edge $i$ is offset}} o_i > 0.
            \]
            Hence, there exists at least one positive offset edge $k$ in C with $o_k > 0$.
            The cycle $C$ therefore induces a cycle in the dependency graph containing a positive offset edge and either an aggregate or an asynchronous hold edge.
            This contradicts the well-formedness condition (3).
        \end{itemize}
    \end{itemize}
\end{proof}

\section{Periodic Well-Formedness}\label{app:period_cycle}
Consider a \FRTLola specification $\Phi$ and some cycle $C = \{(s_1, op, s_2), \dots , (s_n, op, s_1)\}$ in its corresponding dependency graph.\\
First compute the greatest common divider among all pacings \[gcd = GCD(Pace(s_1), Pace(s_2), \dots , Pace(s_n))\] Afterwards, transform the following edges:
\begin{lstlisting}
    output a @pa@ := b.offset(by: off) $\Rightarrow$ 
    output a @gcd@ := b.offset(by:$\frac{pb \cdot off}{gcd}$)
\end{lstlisting}
\begin{lstlisting}
    output a @pa@ := b.hold() $\Rightarrow$
    output a @gcd@ := 
        if Time mod pb$\equiv \frac{pb}{gcd} - 1$ then b.offset(by:$-\frac{pb}{gcd} + 1$) else
        if Time mod pb$\equiv \frac{pb}{gcd} - 1 - gcd$ then b.offset(by:$-\frac{pb}{gcd} + 1 + gcd$) else
        if Time mod pb$\equiv \frac{pb}{gcd} - 1 - 2gcd$ then b.offset(by:$-\frac{pb}{gcd} + 1 + 2gcd$) else
        ...
        if Time mod pa$\equiv$ 0 then b
\end{lstlisting}
\begin{lstlisting}
    output a @pa@ := b.aggregate(over: t, with: $\delta$) $\Rightarrow$
    output a_values @gcd@ := 
        (b.offset(by:$-\lceil\frac{t}{gcd}\rceil + 1$),
        b.offset(by:$-\lceil\frac{t}{gcd}\rceil + 1 + gcd$),
        b.offset(by:$-\lceil\frac{t}{gcd}\rceil + 1 + 2gcd$),
        ...
        b)
    output a @gcd@ := $\delta$(a_values)
\end{lstlisting}

\begin{example}
    Given some \FRTLola specification:
\begin{lstlisting}
    output a @6s@ := b.offset(by: +1).defaults(to:0)
    output b @3s@ := c.aggregate(over: 2s, using: sum)
    output c @2s@ := a.hold(or: 0)
\end{lstlisting}
Now, we translate it according to the construction above:
\begin{lstlisting}
    output a @1s@ := b.offset(by:+3).defaults(to:0)
    output b_values @1s@ := (c.offset(by:-1).defaults(to:0), b)
    output b @1s@ := sum(b_values)
    output c @1s@ := if Time mod 6s $\equiv$ 5 then a.offset(by: -5) else
                        if Time mod 6s $\equiv$ 4 then a.offset(by: -4) else
                        if Time mod 6s $\equiv$ 3 then a.offset(by: -3) else
                        if Time mod 6s $\equiv$ 2 then a.offset(by: -2) else
                        if Time mod 6s $\equiv$ 1 then a.offset(by: -1) else
                        if Time mod 6s $\equiv$ 0 then a
\end{lstlisting}
Hence, we obtain a desugared specification, which now contains a cycle with accumulated edge weight of zero.
This means the specification is not well-formed.
\end{example}

\section{Algorithm Correctness}\label{app:algo_proof}

\begin{definition}[Eval]
\begin{align*}
    \mathit{eval}(\mathit{sync}(x), R, U, M, t_m, t, T) &= \begin{cases} v & \text{if } (x, v, t) \in R \\ ? & \text{otherwise} \end{cases} \\
    \mathit{eval}(\mathit{offset}(x, o), R, U, M, t_m, t, T) &=
        \begin{cases}
        \mathit{eval}(\mathit{sync}(x), \ldots) & \text{if } o = 0 \\
        \bot & \text{if } t \notin [0, T_{\max}]\\
        ? & \text{if } t > t_m \lor (x,t) \in M\\
        \mathit{eval}(\mathit{offset}(x, o'), \ldots, t + 1) & \text{if } o > 0\\
        \mathit{eval}(\mathit{offset}(x, o'), \ldots, t - 1) & \text{if } o < 0\\
        \end{cases}\\
    \mathit{eval}(\mathit{hold}(x), R, U, M, t_m, t, T) &= \begin{cases} \mathit{eval}(\mathit{sync}(x), \ldots) & \text{if } (x, \_, t) \in R \\ \mathit{eval}(\mathit{offset}(x, -1), \ldots) & \text{otherwise} \end{cases}
\end{align*}
with
\[
    o' = \begin{cases}
        o + 1 & \text{if } o < 0 \land \mathit{hasEvent}(x, t) \\
        o - 1 & \text{if } o > 0 \land \mathit{hasEvent}(x, t) \\
        o & \text{otherwise} \end{cases}
\]
and
\[
\mathit{hasEvent}(x, t) = (x, \_, t) \in R \lor (x, t) \in U.
\]
\end{definition}

\subsection{Proof for \Cref{thm:algorithm}}

\begin{proof}
We first argue that the while loop terminates in each iteration of the for loop.
The sets $M$, $U$, and $R$ are finite, and entries move monotonically: each entry can only move from $M$ to $U$ or be discarded, and from $U$ to $R$. No entry can be re-added to $M$ once removed.

\paragraph{Soundness.}
We show by induction on the order in which values are added to $R$ that every $(sr, v, t) \in R$ satisfies $\world(sr)(t) = v$.
For the base case, input stream values are added directly from $I$ in iteration $t$, which by definition agrees with $\world$.
For the inductive step, assume all values currently in $R$ are correct.
A new output stream value $(s, v, t)$ is added to $R$ only when $\mathit{eval}(\varphi_s, R, U, M, t_m, t, T) = v$.
By induction on the expression structure of $\varphi_s$, $\mathit{eval}$ returns a concrete value $v$ only when all dependencies are resolved in $R$.
By the induction hypothesis, all values in $R$ are correct, so the returned value $v$ must equal $\world(s)(t)$ by the uniqueness of $\world$, which follows from the well-definedness of $\varphi$.

\paragraph{Completeness.}
We show that every value $\world(sr)(t) = v$ eventually appears in $R$.
For input streams, values are added in iteration $t$ of the for loop directly from $I$.
For output streams, every $(s, t)$ is initially added to $M$ at the start of iteration $t$.
Since $\world$ is a valid evaluation model, the pacing expression $\varphi^p_s$ evaluates to $\mathit{true}$ at every $t$ where $\world(s)(t) \ne \bot$.
Once all dependencies of $\varphi^p_s$ are resolved in $R$, $\mathit{eval}(\varphi^p_s, \ldots)$ returns $\mathit{true}$, moving $(s, t)$ to $U$.
Subsequently, once all dependencies of $\varphi_s$ are resolved in $R$, $\mathit{eval}(\varphi_s, \ldots)$ returns the concrete value $v = \world(s)(t)$, moving $(s, v, t)$ to $R$.
By well-definedness of $\varphi$, all dependencies are eventually resolved across iterations of the for loop, so every output stream value eventually appears in $R$.
\end{proof}

\section{Delay}
Given some specification $\varphi$ and its corresponding dependency graph $G_\varphi$.
Let $\mathit{weight}(k)$ be the weight of some expression kind $k$, defined as:
\[
\mathit{weight}(k) = \begin{cases}
    o &\text{if } k = (\mathit{Offset}, o)\\
    0 &\text{otherwise}
\end{cases}
\]
The delay $\triangle(s)$ of some stream $s$ is given by:
\[
\triangle(s) = \begin{dcases}
    \infty & \begin{array}{l} \text{if } \exists s': (s, (\_, \mathit{Hold}), s') \in E \\ \wedge\; \triangle(s') > 0 \end{array} \\
    \begin{array}{l} \max\Bigl\{0, \max\bigl\{ \mathit{weight}(k) + \triangle(s') \\ \hspace{2em} \mid (s, (\_, k), s') \in E \bigr\}\Bigr\} \end{array} & \text{otherwise}
\end{dcases}
\]



\begin{theorem}[Delay Soundness]
For every stream $s$ and timestep $t$, the evaluation of $s$ at time $t$ requires at most $\triangle(s)$ future values of any accessed stream to already be present in $R$.
\end{theorem}

\begin{proof}
We distinguish two cases based on whether $\triangle(s)$ is finite or infinite.

\paragraph{Infinite delay:}
$\triangle(s) = \infty$ if $s$ has a Hold edge to a stream $s'$ with $\triangle(s') > 0$, or if $s$ lies on a cycle with positive total weight.
In the Hold case, we cannot bound how long it takes before a new resolved value of $s'$ is produced.

\paragraph{Finite delay:}
If $\triangle(s) < \infty$, all cycles reachable from $s$ have non-positive total weight.
The required number of future values is therefore bounded by the maximum total weight along any path from $s$ in $G_\varphi$, since non-positive cycles do not increase this accumulation.
This maximum is exactly $\triangle(s)$ by definition.
\end{proof}

\section{Proof for \Cref{thm:memory_bound}}\label{app:proof_mem_bound}

\begin{proof}
We show that once there exist $\mu(s)$ newer values of $s$ in $R$, no future call to $\mathit{eval}$ will reference $(s, v, t)$ again.
A value $(s, v, t)$ can only be accessed during the evaluation of some stream $s'$ at event $t'$ that has a direct dependency $(s', (\_, k), s) \in E$, looking up the $k$-th next (or previous) event of $s$ relative to $t'$.
By definition of the delay, the evaluation of $s'$ at event $t'$ requires at most $\triangle(s')$ future events of any accessed stream to be present in $R$.
Hence, $(s, v, t)$ is no longer needed by $s'$ once there are more than $\triangle(s') - k$ newer values of $s$ in $R$, since $s'$ will then access a strictly newer event of $s$.
Taking the maximum over all streams $s'$ and offsets $k$ with $(s', (\_, k), s) \in E$, and adding one to convert the count into a buffer size, yields exactly $\mu_\varphi(s)$.
Therefore, once $\mu(s)$ newer values of $s$ exist in $R$, the value $(s, v, t)$ can be safely removed without affecting any future evaluation.
\end{proof}

\section{Specifications}\label{app:specs}

\subsection{Prediction}
Example for prediction horizon of 3s.
\begin{lstlisting}
import math
input gps_lat : Float64
input gps_lon : Float64
input barometer_pressure : Float64
input barometer_altitude : Float64
input gps_altitude : Float64
output start_lat := start_lat.offset(by:-1).defaults(to: gps_lat)
output start_lon := start_lon.offset(by:-1).defaults(to: gps_lon)
output start_altitude := start_altitude.offset(by:-1).defaults(to: gps_altitude)
output distance_to_start := sqrt((gps_lat-start_lat)*(gps_lat-start_lat)
	+ (gps_lon-start_lon)*(gps_lon-start_lon))*10000.0
output altitude_above_ground := gps_altitude - start_altitude

output distance_to_start_timed @4Hz@ := distance_to_start.hold(or: 0.0)
output outside_geofence := distance_to_start_timed >= 8.0
output predicted_geofence @4Hz@ := distance_to_start_timed.predict(using: linear, over_discrete: 5, in: 3.0s).defaults(to: 0.0) >= 8.0
output geofence_true_positive  :=  outside_geofence &&  predicted_geofence.offset(by: -3s).defaults(to: false)
output geofence_false_positive := !outside_geofence &&  predicted_geofence.offset(by: -3s).defaults(to: false)
output geofence_true_negative  := !outside_geofence && !predicted_geofence.offset(by: -3s).defaults(to: false)
output geofence_false_negative :=  outside_geofence && !predicted_geofence.offset(by: -3s).defaults(to: false)

output altitude_above_ground_timed @4Hz@ := altitude_above_ground.hold(or: 0.0)
output altitude_violation := altitude_above_ground_timed >= 10.0
output predicted_altitude_violation @4Hz@ := altitude_above_ground_timed.predict(using: linear, over_discrete: 5, in: 3.0s).defaults(to: 0.0) >= 10.0
output altitude_true_positive  :=  altitude_violation &&  predicted_altitude_violation.offset(by: -3s).defaults(to: false)
output altitude_false_positive := !altitude_violation &&  predicted_altitude_violation.offset(by: -3s).defaults(to: false)
output altitude_true_negative  := !altitude_violation && !predicted_altitude_violation.offset(by: -3s).defaults(to: false)
output altitude_false_negative :=  altitude_violation && !predicted_altitude_violation.offset(by: -3s).defaults(to: false)
\end{lstlisting}

\subsection{Bounded Specification}

\begin{lstlisting}
import math
input gps_lat : Float64
input gps_lon : Float64
input barometer_pressure : Float64
input barometer_altitude : Float64
input gps_altitude : Float64

output start_lat := start_lat.offset(by:-1).defaults(to: gps_lat)
output start_lon := start_lon.offset(by:-1).defaults(to: gps_lon)
output start_altitude := start_altitude.offset(by:-1).defaults(to: gps_altitude)

output distance_to_start := sqrt((gps_lat-start_lat)*(gps_lat-start_lat)
    + (gps_lon-start_lon)*(gps_lon-start_lon))*10000.0
output altitude_above_ground := gps_altitude - start_altitude

output distance_to_start_timed @1Hz@ := distance_to_start.hold(or: 0.0)

output outside_geofence := distance_to_start_timed >= 8.0

trigger !(outside_geofence -> !outside_geofence.offset(by: x).defaults(to: false))
\end{lstlisting}

%% file: bibliography.bib
@inproceedings{DBLP:conf/time/DAngeloSSRFSMM05,
  author       = {Ben D'Angelo and
                  Sriram Sankaranarayanan and
                  C{\'{e}}sar S{\'{a}}nchez and
                  Will Robinson and
                  Bernd Finkbeiner and
                  Henny B. Sipma and
                  Sandeep Mehrotra and
                  Zohar Manna},
  title        = {{LOLA:} Runtime Monitoring of Synchronous Systems},
  booktitle    = {12th International Symposium on Temporal Representation and Reasoning
                  {(TIME} 2005), 23-25 June 2005, Burlington, Vermont, {USA}},
  pages        = {166--174},
  publisher    = {{IEEE} Computer Society},
  year         = {2005},
  doi          = {10.1109/TIME.2005.26},
  bibsource    = {dblp computer science bibliography, https://dblp.org}
}

@inproceedings{DBLP:conf/fm/BaumeisterFKS24,
  author       = {Jan Baumeister and
                  Bernd Finkbeiner and
                  Florian Kohn and
                  Frederik Scheerer},
  editor       = {Andr{\'{e}} Platzer and
                  Kristin Yvonne Rozier and
                  Matteo Pradella and
                  Matteo Rossi},
  title        = {A Tutorial on Stream-Based Monitoring},
  booktitle    = {Formal Methods - 26th International Symposium, {FM} 2024, Milan, Italy,
                  September 9-13, 2024, Proceedings, Part {II}},
  series       = {Lecture Notes in Computer Science},
  pages        = {624--648},
  publisher    = {Springer},
  year         = {2024},
  doi          = {10.1007/978-3-031-71177-0\_33},
  bibsource    = {dblp computer science bibliography, https://dblp.org}
}

@inproceedings{DBLP:conf/cav/BaumeisterCFS25,
  author       = {Jan Baumeister and
                  Arthur Correnson and
                  Bernd Finkbeiner and
                  Frederik Scheerer},
  editor       = {Ruzica Piskac and
                  Zvonimir Rakamaric},
  title        = {An Intermediate Program Representation for Optimizing Stream-Based
                  Languages},
  booktitle    = {Computer Aided Verification - 37th International Conference, {CAV}
                  2025, Zagreb, Croatia, July 23-25, 2025, Proceedings, Part {III}},
  series       = {Lecture Notes in Computer Science},
  volume       = {15933},
  pages        = {393--407},
  publisher    = {Springer},
  year         = {2025},
  doi          = {10.1007/978-3-031-98682-6\_20},
  bibsource    = {dblp computer science bibliography, https://dblp.org}
}

@phdthesis{DBLP:phd/dnb/Schwenger22,
  author       = {Maximilian Schwenger},
  title        = {Statically-analyzed stream monitoring for cyber-physical Systems},
  school       = {Saarland University, Saarbr{\"{u}}cken, Germany},
  year         = {2022},
  url          = {https://publikationen.sulb.uni-saarland.de/handle/20.500.11880/33890},
  urn          = {urn:nbn:de:bsz:291--ds-370140},
  bibsource    = {dblp computer science bibliography, https://dblp.org}
}

@inproceedings{DBLP:conf/sbmf/ConventHLS0T18,
  author       = {Lukas Convent and
                  Sebastian Hungerecker and
                  Martin Leucker and
                  Torben Scheffel and
                  Malte Schmitz and
                  Daniel Thoma},
  editor       = {Tiago Massoni and
                  Mohammad Reza Mousavi},
  title        = {{TeSSLa}: Temporal Stream-Based Specification Language},
  booktitle    = {Formal Methods: Foundations and Applications - 21st Brazilian Symposium,
                  {SBMF} 2018, Salvador, Brazil, November 26-30, 2018, Proceedings},
  series       = {Lecture Notes in Computer Science},
  pages        = {144--162},
  publisher    = {Springer},
  year         = {2018},
  doi          = {10.1007/978-3-030-03044-5\_10},
  bibsource    = {dblp computer science bibliography, https://dblp.org}
}

@phdthesis{DBLP:phd/dnb/Scheffel22,
  author       = {Torben Scheffel},
  title        = {Expressiveness and complexity of stream-based specification languages},
  school       = {University of L{\"{u}}beck, Germany},
  year         = {2022},
  url          = {https://www.zhb.uni-luebeck.de/epubs/ediss2525.pdf},
  urn          = {urn:nbn:de:gbv:841-20210907283},
  bibsource    = {dblp computer science bibliography, https://dblp.org}
}

@inproceedings{DBLP:conf/rv/GorostiagaS18,
  author       = {Felipe Gorostiaga and
                  C{\'{e}}sar S{\'{a}}nchez},
  editor       = {Christian Colombo and
                  Martin Leucker},
  title        = {Striver: Stream Runtime Verification for Real-Time Event-Streams},
  booktitle    = {Runtime Verification - 18th International Conference, {RV} 2018, Limassol,
                  Cyprus, November 10-13, 2018, Proceedings},
  series       = {Lecture Notes in Computer Science},
  pages        = {282--298},
  publisher    = {Springer},
  year         = {2018},
  doi          = {10.1007/978-3-030-03769-7\_16},
  bibsource    = {dblp computer science bibliography, https://dblp.org}
}

@article{DBLP:journals/sttt/GorostiagaS21,
  author       = {Felipe Gorostiaga and
                  C{\'{e}}sar S{\'{a}}nchez},
  title        = {Stream runtime verification of real-time event streams with the Striver
                  language},
  journal      = {Int. J. Softw. Tools Technol. Transf.},
  volume       = {23},
  number       = {2},
  pages        = {157--183},
  year         = {2021},
  doi          = {10.1007/S10009-021-00605-3},
  bibsource    = {dblp computer science bibliography, https://dblp.org}
}

@inproceedings{DBLP:conf/iros/ChouY020,
  author       = {Yi Chou and
                  Hansol Yoon and
                  Sriram Sankaranarayanan},
  title        = {Predictive Runtime Monitoring of Vehicle Models Using Bayesian Estimation
                  and Reachability Analysis},
  booktitle    = {{IEEE/RSJ} International Conference on Intelligent Robots and Systems,
                  {IROS} 2020, Las Vegas, NV, USA, October 24, 2020 - January 24, 2021},
  pages        = {2111--2118},
  publisher    = {{IEEE}},
  year         = {2020},
  doi          = {10.1109/IROS45743.2020.9340755},
  bibsource    = {dblp computer science bibliography, https://dblp.org}
}

@inproceedings{DBLP:conf/rv/CairoliBP21,
  author       = {Francesca Cairoli and
                  Luca Bortolussi and
                  Nicola Paoletti},
  editor       = {Lu Feng and
                  Dana Fisman},
  title        = {Neural Predictive Monitoring Under Partial Observability},
  booktitle    = {Runtime Verification - 21st International Conference, {RV} 2021, Virtual
                  Event, October 11-14, 2021, Proceedings},
  series       = {Lecture Notes in Computer Science},
  pages        = {121--141},
  publisher    = {Springer},
  year         = {2021},
  doi          = {10.1007/978-3-030-88494-9\_7},
  bibsource    = {dblp computer science bibliography, https://dblp.org}
}

@article{DBLP:journals/inffus/ChenMLWL23,
  author       = {Zonglei Chen and
                  Minbo Ma and
                  Tianrui Li and
                  Hongjun Wang and
                  Chongshou Li},
  title        = {Long sequence time-series forecasting with deep learning: {A} survey},
  journal      = {Inf. Fusion},
  volume       = {97},
  pages        = {101819},
  year         = {2023},
  doi          = {10.1016/J.INFFUS.2023.101819},
  bibsource    = {dblp computer science bibliography, https://dblp.org}
}

@article{Spinger:article/bmc/FilipowMTRDDLS23,
    author     = {Nicole Filipow and
                  Eleanor Main and
                  Gizem Tanriver and
                  Emma Raywood and
                  Gwyneth Davies and
                  Helen Douglas and
                  Aidan Laverty and
                  Sanja Stanojevic},
    title = {Exploring flexible polynomial regression as a method to align routine clinical outcomes with daily data capture through remote technologies},
    journal = {BMC Med Res Methodol},
    year = {2023},
    doi = {10.1186/s12874-023-01942-4}
}

@inproceedings{DBLP:conf/rv/FinkbeinerOPS20,
  author       = {Bernd Finkbeiner and
                  Stefan Oswald and
                  Noemi Passing and
                  Maximilian Schwenger},
  editor       = {Jyotirmoy Deshmukh and
                  Dejan Nickovic},
  title        = {Verified Rust Monitors for Lola Specifications},
  booktitle    = {Runtime Verification - 20th International Conference, {RV} 2020, Los
                  Angeles, CA, USA, October 6-9, 2020, Proceedings},
  series       = {Lecture Notes in Computer Science},
  pages        = {431--450},
  publisher    = {Springer},
  year         = {2020},
  doi          = {10.1007/978-3-030-60508-7\_24},
  bibsource    = {dblp computer science bibliography, https://dblp.org}
}

@inproceedings{DBLP:conf/rv/BozzelliS14,
  author       = {Laura Bozzelli and
                  C{\'{e}}sar S{\'{a}}nchez},
  editor       = {Borzoo Bonakdarpour and
                  Scott A. Smolka},
  title        = {Foundations of Boolean Stream Runtime Verification},
  booktitle    = {Runtime Verification - 5th International Conference, {RV} 2014, Toronto,
                  ON, Canada, September 22-25, 2014. Proceedings},
  series       = {Lecture Notes in Computer Science},
  pages        = {64--79},
  publisher    = {Springer},
  year         = {2014},
  doi          = {10.1007/978-3-319-11164-3\_6},
  bibsource    = {dblp computer science bibliography, https://dblp.org}
}

@inproceedings{DBLP:conf/atva/RaszykBKT19,
  author       = {Martin Raszyk and
                  David A. Basin and
                  Srdan Krstic and
                  Dmitriy Traytel},
  editor       = {Yu{-}Fang Chen and
                  Chih{-}Hong Cheng and
                  Javier Esparza},
  title        = {Multi-head Monitoring of Metric Temporal Logic},
  booktitle    = {Automated Technology for Verification and Analysis - 17th International
                  Symposium, {ATVA} 2019, Taipei, Taiwan, October 28-31, 2019, Proceedings},
  series       = {Lecture Notes in Computer Science},
  pages        = {151--170},
  publisher    = {Springer},
  year         = {2019},
  doi          = {10.1007/978-3-030-31784-3\_9},
  bibsource    = {dblp computer science bibliography, https://dblp.org}
}

@article{DBLP:journals/tosem/BauerLS11,
  author       = {Andreas Bauer and
                  Martin Leucker and
                  Christian Schallhart},
  title        = {Runtime Verification for {LTL} and {TLTL}},
  journal      = {{ACM} Trans. Softw. Eng. Methodol.},
  volume       = {20},
  number       = {4},
  pages        = {14:1--14:64},
  year         = {2011},
  doi          = {10.1145/2000799.2000800},
  bibsource    = {dblp computer science bibliography, https://dblp.org}
}

@inproceedings{DBLP:conf/iccps/LindemannQDP23,
  author       = {Lars Lindemann and
                  Xin Qin and
                  Jyotirmoy V. Deshmukh and
                  George J. Pappas},
  editor       = {Sayan Mitra and
                  Nalini Venkatasubramanian and
                  Abhishek Dubey and
                  Lu Feng and
                  Mahsa Ghasemi and
                  Jonathan Sprinkle},
  title        = {Conformal Prediction for {STL} Runtime Verification},
  booktitle    = {Proceedings of the {ACM/IEEE} 14th International Conference on Cyber-Physical
                  Systems, {ICCPS} 2023, (with CPS-IoT Week 2023), San Antonio, TX,
                  USA, May 9-12, 2023},
  pages        = {142--153},
  publisher    = {{ACM}},
  year         = {2023},
  doi          = {10.1145/3576841.3585927},
  bibsource    = {dblp computer science bibliography, https://dblp.org}
}

@inproceedings{DBLP:conf/rv/Leucker12,
  author       = {Martin Leucker},
  editor       = {Shaz Qadeer and
                  Serdar Tasiran},
  title        = {Sliding between Model Checking and Runtime Verification},
  booktitle    = {Runtime Verification, Third International Conference, {RV} 2012, Istanbul,
                  Turkey, September 25-28, 2012, Revised Selected Papers},
  series       = {Lecture Notes in Computer Science},
  pages        = {82--87},
  publisher    = {Springer},
  year         = {2012},
  doi          = {10.1007/978-3-642-35632-2\_10},
  bibsource    = {dblp computer science bibliography, https://dblp.org}
}

@inproceedings{DBLP:conf/nfm/ZhangLD12,
  author       = {Xian Zhang and
                  Martin Leucker and
                  Wei Dong},
  editor       = {Alwyn Goodloe and
                  Suzette Person},
  title        = {Runtime Verification with Predictive Semantics},
  booktitle    = {{NASA} Formal Methods - 4th International Symposium, {NFM} 2012, Norfolk,
                  VA, USA, April 3-5, 2012. Proceedings},
  series       = {Lecture Notes in Computer Science},
  pages        = {418--432},
  publisher    = {Springer},
  year         = {2012},
  doi          = {10.1007/978-3-642-28891-3\_37},
  bibsource    = {dblp computer science bibliography, https://dblp.org}
}

@article{DBLP:journals/jacm/BasinKMZ15,
  author       = {David A. Basin and
                  Felix Klaedtke and
                  Samuel M{\"{u}}ller and
                  Eugen Zalinescu},
  title        = {Monitoring Metric First-Order Temporal Properties},
  journal      = {J. {ACM}},
  volume       = {62},
  number       = {2},
  pages        = {15:1--15:45},
  year         = {2015},
  doi          = {10.1145/2699444},
  bibsource    = {dblp computer science bibliography, https://dblp.org}
}

@inproceedings{DBLP:conf/rv/SchumannMR15,
  author       = {Johann Schumann and
                  Patrick Moosbrugger and
                  Kristin Y. Rozier},
  editor       = {Ezio Bartocci and
                  Rupak Majumdar},
  title        = {{R2U2:} Monitoring and Diagnosis of Security Threats for Unmanned
                  Aerial Systems},
  booktitle    = {Runtime Verification - 6th International Conference, {RV} 2015 Vienna,
                  Austria, September 22-25, 2015. Proceedings},
  series       = {Lecture Notes in Computer Science},
  pages        = {233--249},
  publisher    = {Springer},
  year         = {2015},
  doi          = {10.1007/978-3-319-23820-3\_15},
  bibsource    = {dblp computer science bibliography, https://dblp.org}
}

@inproceedings{DBLP:conf/tacas/BaumeisterFSSW25,
  author       = {Jan Baumeister and
                  Bernd Finkbeiner and
                  Frederik Scheerer and
                  Julian Siber and
                  Tobias Wagenpfeil},
  editor       = {Arie Gurfinkel and
                  Marijn Heule},
  title        = {Stream-Based Monitoring of Algorithmic Fairness},
  booktitle    = {Tools and Algorithms for the Construction and Analysis of Systems
                  - 31st International Conference, {TACAS} 2025, Held as Part of the
                  International Joint Conferences on Theory and Practice of Software,
                  {ETAPS} 2025, Hamilton, ON, Canada, May 3-8, 2025, Proceedings, Part
                  {I}},
  series       = {Lecture Notes in Computer Science},
  pages        = {60--81},
  publisher    = {Springer},
  year         = {2025},
  doi          = {10.1007/978-3-031-90643-5\_4},
  bibsource    = {dblp computer science bibliography, https://dblp.org}
}

@inproceedings{DBLP:conf/lop/LichtensteinPZ85,
  author       = {Orna Lichtenstein and
                  Amir Pnueli and
                  Lenore D. Zuck},
  editor       = {Rohit Parikh},
  title        = {The Glory of the Past},
  booktitle    = {Logics of Programs, Conference, Brooklyn College, New York, NY, USA,
                  June 17-19, 1985, Proceedings},
  series       = {Lecture Notes in Computer Science},
  pages        = {196--218},
  publisher    = {Springer},
  year         = {1985},
  doi          = {10.1007/3-540-15648-8\_16},
  bibsource    = {dblp computer science bibliography, https://dblp.org}
}

@inproceedings{DBLP:conf/time/ArtaleGGMM23,
  author       = {Alessandro Artale and
                  Luca Geatti and
                  Nicola Gigante and
                  Andrea Mazzullo and
                  Angelo Montanari},
  editor       = {Alexander Artikis and
                  Florian Bruse and
                  Luke Hunsberger},
  title        = {{LTL} over Finite Words Can Be Exponentially More Succinct Than Pure-Past
                  LTL, and vice versa},
  booktitle    = {30th International Symposium on Temporal Representation and Reasoning,
                  {TIME} 2023, {NCSR} Demokritos, Athens, Greece, September 25-26, 2023},
  series       = {LIPIcs},
  pages        = {2:1--2:14},
  publisher    = {Schloss Dagstuhl - Leibniz-Zentrum f{\"{u}}r Informatik},
  year         = {2023},
  doi          = {10.4230/LIPICS.TIME.2023.2},
  bibsource    = {dblp computer science bibliography, https://dblp.org}
}

@inproceedings{DBLP:conf/tacas/HavelundR02,
  author       = {Klaus Havelund and
                  Grigore Rosu},
  editor       = {Joost{-}Pieter Katoen and
                  Perdita Stevens},
  title        = {Synthesizing Monitors for Safety Properties},
  booktitle    = {Tools and Algorithms for the Construction and Analysis of Systems,
                  8th International Conference, {TACAS} 2002, Held as Part of the Joint
                  European Conference on Theory and Practice of Software, {ETAPS} 2002,
                  Grenoble, France, April 8-12, 2002, Proceedings},
  series       = {Lecture Notes in Computer Science},
  pages        = {342--356},
  publisher    = {Springer},
  year         = {2002},
  doi          = {10.1007/3-540-46002-0\_24},
  bibsource    = {dblp computer science bibliography, https://dblp.org}
}

@inproceedings{DBLP:conf/rv/BaumeisterFS25,
  author       = {Jan Baumeister and
                  Bernd Finkbeiner and
                  Frederik Scheerer},
  editor       = {Bettina K{\"{o}}nighofer and
                  Hazem Torfah},
  title        = {Active Monitoring with {RTLola}: {A} Specification-Guided Scheduling
                  Approach},
  booktitle    = {Runtime Verification - 25th International Conference, {RV} 2025, Graz,
                  Austria, September 15-19, 2025, Proceedings},
  series       = {Lecture Notes in Computer Science},
  pages        = {181--201},
  publisher    = {Springer},
  year         = {2025},
  doi          = {10.1007/978-3-032-05435-7\_11},
  bibsource    = {dblp computer science bibliography, https://dblp.org}
}

@article{baumeister2025intermediate,
  title={An Intermediate Program Representation for Optimizing Stream-Based Languages}, 
  author={Jan Baumeister and Arthur Correnson and Bernd Finkbeiner and Frederik Scheerer},
  year={2025},
  eprint={2504.21458},
  archivePrefix={arXiv},
  primaryClass={cs.LO},
  url={https://arxiv.org/abs/2504.21458}, 
}

@inproceedings{DBLP:conf/cav/HiplerKLS24,
  author       = {Raik Hipler and
                  Hannes Kallwies and
                  Martin Leucker and
                  C{\'{e}}sar S{\'{a}}nchez},
  editor       = {Arie Gurfinkel and
                  Vijay Ganesh},
  title        = {General Anticipatory Runtime Verification},
  booktitle    = {Computer Aided Verification - 36th International Conference, {CAV}
                  2024, Montreal, QC, Canada, July 24-27, 2024, Proceedings, Part {II}},
  series       = {Lecture Notes in Computer Science},
  volume       = {14682},
  pages        = {133--155},
  publisher    = {Springer},
  year         = {2024},
  doi          = {10.1007/978-3-031-65630-9\_7},
  bibsource    = {dblp computer science bibliography, https://dblp.org}
}
